\documentclass[11pt]{article}
\usepackage[T1]{fontenc}
\usepackage[utf8]{inputenc}
\usepackage{hyperref}
\hypersetup{colorlinks,allcolors=blue}
\usepackage{appendix}
\usepackage{amsmath}
\usepackage{algorithm}
\usepackage[noend]{algpseudocode}
\usepackage{xcolor}
\usepackage{listings}
\usepackage{graphicx}
\usepackage{amssymb, amsthm, amsfonts, graphicx, color, subcaption, enumerate, bm, enumitem, array, mathtools}
\setlist{itemsep=1pt}

\usepackage{thmtools, thm-restate}
\usepackage{mathdots}

\usepackage{tikz}
\usetikzlibrary{arrows.meta}
\usepackage{csquotes}
\usepackage{enumitem}
\usepackage[style=trad-alpha,natbib=true,maxcitenames=4]{biblatex}
\usepackage{lastpage}
\usepackage[margin=1in]{geometry}

\newcommand{\dist}{\operatorname{dist}}

\newcommand{\CONGEST}{\mathsf{CONGEST}}
\newcommand{\ID}{\operatorname{ID}}
\newcommand{\poly}{\operatorname{poly}}

\newcommand{\Tsssp}{T_{\mathsf{SSSP}}}

\newcommand{\rpath}{RPaths}
\newcommand{\apxrpath}{Apx-\rpath{}}
\newcommand{\secsisp}{2-SiSP}
\newcommand{\stpath}{(s=v_0, v_1, v_2, \ldots, v_{h_{st}}=t)}

\usepackage[colorinlistoftodos,prependcaption,textsize=tiny,disable]{todonotes}
\newcommand{\yanyu}[1]{\todo[backgroundcolor=red!25]{Yanyu: #1}}

\newcommand{\yijun}[1]{\todo[backgroundcolor=cyan!25]{Yijun: #1}}

\title{Optimal Distributed Replacement Paths}
 \author{Yi-Jun Chang\footnote{National University of Singapore. ORCID: 0000-0002-0109-2432. Email: cyijun@nus.edu.sg} 
 \and Yanyu Chen\footnote{National University of Singapore. ORCID: 0009-0008-8068-1649. Email: yanyu.chen@u.nus.edu}
 \and Dipan Dey\footnote{Tata Institute of Fundamental Research, India. ORCID: 0009-0001-0675-8790. Email: dipan.dey@tifr.res.in}
 \and Gopinath Mishra\footnote{National University of Singapore. ORCID: 0000-0003-0540-0292. Email: gopinath@nus.edu.sg}  
 \and Hung Thuan Nguyen\footnote{National University of Singapore. ORCID: 0009-0006-7993-2952.  Email: hung@u.nus.edu}
 \and Bryce Sanchez\footnote{National University of Singapore. ORCID: 0009-0000-5807-297X. Email: bryce\_sanchez\_21@u.nus.edu} }
\date{}

\usepackage{cleveref}

\newtheorem{lemma}{Lemma}[section]

\newtheorem{definition}[lemma]{Definition}
\newtheorem{corollary}[lemma]{Corollary}
\newtheorem{observation}[lemma]{Observation}
\newtheorem{proposition}[lemma]{Proposition}

\begin{document}

\maketitle

\begin{abstract}
We study the \emph{replacement paths} problem in the $\CONGEST$ model of distributed computing.  Given an $s$-$t$ shortest path $P$, the goal is to compute, for every edge $e$ in $P$, the shortest-path distance from $s$ to $t$ avoiding $e$. For \emph{unweighted} directed graphs, we establish the \emph{tight} randomized round complexity bound for this problem as $\widetilde{\Theta}(n^{2/3} + D)$ by showing matching upper and lower bounds. Our upper bound extends to $(1+\epsilon)$-approximation for \emph{weighted} directed graphs. 
Our lower bound applies even to the \emph{second simple shortest path} problem, which asks only for the smallest replacement path length.
These results improve upon the very recent work of Manoharan and Ramachandran (SIROCCO 2024), who showed a lower bound of $\widetilde{\Omega}(n^{1/2} + D)$ and an upper bound of $\widetilde{O}(n^{2/3} + \sqrt{n h_{st}} + D)$, where $h_{st}$ is the number of hops in the given $s$-$t$ shortest path $P$.  

\yijun{make sure to hide all todo-comments before arxiv submission}
\yijun{remember to change node/nodes to vertex/vertices}
\end{abstract}

\thispagestyle{empty}
\newpage
\thispagestyle{empty}
\tableofcontents
\newpage
\pagenumbering{arabic}

\section{Introduction}
We study the \emph{Replacement Paths (\rpath{})} problem.  Given an $s$-$t$ shortest path $P$, the goal is to compute, for every edge $e$ in $P$, the shortest-path distance from $s$ to $t$ avoiding $e$. This problem has applications in network reliability, transportation systems, and distributed computing, where efficiently rerouting traffic in response to edge failures is critical.

We focus on the $\CONGEST$ model of distributed computing~\cite{peleg2000distributed}, where the network is modeled as a graph $G=(V,E)$. The communication proceeds in synchronous rounds. In each round, each vertex can send each neighbor an $O(\log n)$-bit message.
Shortest paths computation plays a critical role in distributed computing. Many classical shortest paths problems, and their related problems, are well-studied in the $\CONGEST$ model of distributed computing: \emph{Single-Source Shortest Path (SSSP)}~\cite{ashvinkumar2024parallel,elkin2020distributed, cao2023parallel,chechik2020single,forster2018faster,ghaffari2018improved,Rozhon2023}, \emph{All-Pairs Shortest Paths (APSP)}~\cite{agarwal2018deterministic,agarwal2020faster,bernstein2019distributed,holzer2012optimal,huang2017distributed,lenzen2015fast,nanongkai2014distributed}, and \emph{reachability}~\cite{chechik_et_al:LIPIcs.DISC.2019.11,jambulapati2019parallel}.

The \rpath{} problem, as well as its many variants, has been extensively studied in the centralized setting~\cite{Roditty2012a,hershberger2003difficulty,hershberger2001vickrey,NardelliPW03,MalikMG89,weimann2013replacement,grandoni2019faster,chechik2019near,gupta2020multiple,grandoni2012improved,gu2021faster,bilo2021near,DeyGupta22,bhosle2004replacement,williams2018subcubic,williams2022algorithms}. Surprisingly, despite having applications in fault-tolerant distributed computing, the \rpath{} problem has not received much attention in distributed computing yet. 
Ghaffari and Parter~\cite{ghaffari2016near} showed an $O(D \log n)$-round algorithm for single-source \rpath{} in unweighted undirected graphs, where $n=|V|$ is the number of vertices and $D$ is the diameter of the graph $G=(V,E)$. The first systematic study of the \rpath{} problem was done very recently by \citet{manoharan2024computing}. They consider both directed and undirected graphs, weighted and unweighted graphs. In particular, they obtained a \emph{tight} bound $\widetilde{\Theta}(n)$ for the \rpath{} problem in \emph{weighted directed} graphs by showing matching upper and lower bounds.\footnote{The notations $\widetilde{O}(\cdot)$, $\widetilde{\Omega}(\cdot)$, and $\widetilde{\Theta}(\cdot)$ suppress any $1/\poly(\log n)$ and $\poly(\log n)$ factor.}



However, the complexities of the \rpath{} problem become more complicated in some other cases.
\citet{manoharan2024computing} presented an algorithm that solves the \rpath{} problem on \emph{unweighted directed} graphs in $\widetilde{O}(n^{2/3} + \sqrt{n h_{st}} + D)$ rounds,  where $h_{st}$ is the number of hops
in the given shortest path $P$ from $s$ to $t$. They complemented their upper bound with a lower bound of $\widetilde{\Omega}\left(\sqrt{n} + D\right)$, which holds even when $h_{st}$ and $D$ are as small as $O(\log n)$. 
Given the large gap between the upper and lower bounds, as mentioned by \citet{manoharan2024computing}, it is natural to investigate the following questions:

\begin{center}
\begin{description}
    \item[(Q1)] Is it possible to narrow or close the gap between the upper bound $\widetilde{O}(n^{2/3} + \sqrt{n h_{st}} + D)$ and the lower bound $\widetilde{\Omega}\left(\sqrt{n} + D\right)$?
    \item[(Q2)] Is it possible to reduce or eliminate the dependence on $h_{st}$?
\end{description}    
\end{center}

Beyond \emph{unweighted directed} graphs, these two questions also apply to other cases of the \rpath{} problem. \citet{manoharan2024computing} also studied the \emph{Approximate Replacement Paths (\apxrpath{})} problem. In this setting, given an $s$-$t$ shortest path $P$, the goal is to compute, for every edge $e$ in $P$, an $(1+\epsilon)$-approximation of the shortest-path distance from $s$ to $t$ avoiding $e$. \citeauthor*{manoharan2024computing} showed that \apxrpath{} on \emph{weighted directed graphs} also admit the same upper bound $\widetilde{O}(n^{2/3} + \sqrt{n h_{st}} + D)$ and lower bound $\widetilde{\Omega}\left(\sqrt{n} + D\right)$ for any constant $\epsilon \in (0, 1)$. 


\subsection{Our Contribution}
We answer both questions (Q1) and (Q2) affirmatively. For \emph{unweighted directed} graphs, we establish the \emph{tight} randomized round complexity bound for the \rpath{} problem as $\widetilde{\Theta}(n^{2/3} + D)$ by showing matching upper and lower bounds. Throughout the paper, we say that an algorithm succeeds \emph{with high probability} if it succeeds with probability $1 - 1/\poly(n)$.

\begin{restatable}[Upper bound for unweighted directed \rpath{}]{thm}{mainUB}
    \label{thm:main_UB}
    There exists an $\widetilde{O}(n^{2/3}+D)$-round randomized algorithm that solves \rpath{} in unweighted directed graphs with high probability.
\end{restatable}

Our lower bound applies even to the \emph{Second Simple Shortest Path (\secsisp{})} problem, which asks only for the smallest replacement path length, over all choices of $e$ in $P$. In other words, the goal is to find a shortest path from $s$ to $t$ that avoids at least one edge in the given shortest path $P$.

\begin{restatable}[Lower bound for unweighted directed \rpath{}]{thm}{mainLB}
    \label{thm:main_LB}
    Any algorithm that solves \rpath{} or \secsisp{} with constant probability in unweighted directed graphs requires $\widetilde{\Omega}(n^{2/3}+D)$ rounds.
\end{restatable}

\Cref{thm:main_LB} is a simplified statement of our lower bound. In our actual lower bound proof, we demonstrate an infinite sequence of graphs with $D = \Theta(\log n)$ in which \secsisp{} requires $\widetilde{\Omega}(n^{2/3})$ rounds to solve. Our lower bound is established using the framework of the work~\cite{das2011distributed}. This framework has been extremely useful in proving tight lower bounds of $\widetilde{\Omega}(\sqrt{n})$ for most problems in the complexity class $\widetilde{\Theta}(\sqrt{n})$. 

Our results place unweighted directed \rpath{} and \secsisp{} among the rare problems whose complexity is strictly between $o(n)$ and $\omega(\sqrt{n})$. While this result may seem intuitive for \rpath{}, given that it involves reporting multiple distances, it is more surprising for \secsisp{}, which returns only a single distance value.

Similar intermediate complexities arise in the area of \emph{distributed subgraph finding}. In particular, the problem of listing all $k$-cliques has a tight round complexity bound of $\widetilde{\Theta}(n^{1-2/k})$~\cite{censor2021tight}. However, these subgraph finding problems differ fundamentally from \rpath{} and other shortest path problems: Distributed subgraph finding is inherently \emph{local}, meaning it can be solved in $O(1)$ rounds given unlimited bandwidth. 

It was previously known that the techniques from the work~\cite{das2011distributed} can be used to obtain lower bounds of complexity other than $\Omega(\sqrt{n})$. Specifically, computing $k$-source shortest paths in an unweighted directed graph is known to admit a lower bound of $\widetilde{\Omega}(\sqrt{nk} + D)$, which spans the spectrum between $\sqrt{n}$ and $n$ as the number $k$ of sources varies~\cite{manoharan2024distributed}. The lower bound is known to be \emph{tight} when $k \geq n^{1/3}$~\cite{manoharan2024computingMinWeightCycleCONGEST}.


Our upper bound extends to $(1+\epsilon)$-\apxrpath{} for \emph{weighted directed} graphs. Throughout the paper, we assume that in weighted graphs, all edge weights are positive integers bounded by some polynomial in $n$.

\begin{restatable}[Upper bound for weighted directed $(1+\epsilon)$-\apxrpath{}]{thm}{apxUB}
    \label{thm:apx_UB}
    For any constant $\epsilon \in (0,1)$, there exists an $\widetilde{O}(n^{2/3}+D)$-round randomized algorithm that solves $(1+\epsilon)$-\apxrpath{} in weighted directed graphs with high probability.
\end{restatable}

Our lower bound proof inherently does not apply to approximation algorithms, so closing the gap between our upper bound of $\widetilde{O}(n^{2/3} + D)$ and the lower bound of $\widetilde{\Omega}(\sqrt{n} + D)$ by \citet{manoharan2024computing} remains an intriguing open question.

We summarize our new results, along with a comparison with previous work, in \Cref{table:main}.

\begin{table}[ht!]
\centering
\begin{tabular}{lllll}
 & \multicolumn{2}{l}{\textbf{Upper bounds}} & \multicolumn{2}{l}{\textbf{Lower bounds}} \\ \hline
\multicolumn{1}{|l|}{\bf\rpath{}} & $\widetilde{O}(n^{2/3} + \sqrt{n h_{st}} + D)$ & \multicolumn{1}{l|}{\cite{manoharan2024computing}} & $\widetilde{\Omega}(\sqrt{n}+D)$ & \multicolumn{1}{l|}{\cite{manoharan2024computing}} \\ \cline{2-5} 
\multicolumn{1}{|l|}{\small Unweighted directed graphs} & $\widetilde{O}(n^{2/3}+D)$ & \multicolumn{1}{l|}{\Cref{thm:main_UB}} & $\widetilde{\Omega}(n^{2/3}+D)$ & \multicolumn{1}{l|}{\Cref{thm:main_LB}} \\ \hline
\multicolumn{1}{|l|}{\bf \apxrpath{}} & $\widetilde{O}(n^{2/3} + \sqrt{n h_{st}} + D)$ & \multicolumn{1}{l|}{\cite{manoharan2024computing}} & $\widetilde{\Omega}(\sqrt{n}+D)$ & \multicolumn{1}{l|}{\cite{manoharan2024computing}} \\ \cline{2-3}
\multicolumn{1}{|l|}{\small Weighted directed graphs} & $\widetilde{O}(n^{2/3}+D)$ & \multicolumn{1}{l|}{\Cref{thm:apx_UB}} &  & \multicolumn{1}{l|}{} \\ \hline
\end{tabular}
\caption{Summary of our results and comparison with previous work.}\label{table:main}
\end{table}

\paragraph{Remark.} While we establish that $\widetilde{\Theta}(n^{2/3} + D)$ is a tight bound for the unweighted undirected \rpath{} problem in terms of $n$ and $D$, the trivial $O(h_{st} \cdot \Tsssp)$-round algorithm---which computes SSSP in $G \setminus e$ for each edge $e$ in $P$---is faster than our $\widetilde{O}(n^{2/3} + D)$-round algorithm when $h_{st}$ is small. Here $\Tsssp = \widetilde{O}(\sqrt{n}+D) + n^{(2/5) + o(1)} \cdot O(D^{2/5})$ is the round complexity of the SSSP problem in unweighted directed graphs~\cite{cao2023parallel}. Our lower bound construction requires $h_{st} = \Theta(n^{2/3})$, whereas the previous $\widetilde{\Omega}(\sqrt{n}+D)$ lower bound of~\citet{manoharan2024computing} applies even for constant $h_{st}$.

\subsection{Additional Related Work}

The \rpath{} problem has a long line of research in the centralized setting. Hershberger and Suri \cite{hershberger2001vickrey} first introduced this problem and solved it for weighted undirected graphs in $O(m+n \log n)$ time, where $m = |E|$ is the number of edges in the input graph $G=(V,E)$. Later, Hershberger, Suri, and Bhosle~\cite{hershberger2003difficulty} provided a lower bound of $\Omega(m \sqrt{n})$ time in the ``path comparison model'' for unweighted directed graphs with $m=O(n \sqrt{n})$. Roditty and Zwick~\cite{Roditty2012a} designed a Monte Carlo algorithm with running time $\widetilde{O}(m \sqrt{n})$ for unweighted directed graphs, matching the lower bound. The classical algorithm by Yen~\cite{yen1971finding} solves the weighted directed \rpath{} problem in $\widetilde{O}(mn)$ time, which is \emph{conditionally} optimal assuming the \emph{MWC conjecture} proposed by \citet{AgarwalR18}.

A more general version of the \rpath{} problem considers all possible targets $t \in V\setminus\{s\}$ in the graph after fixing a source $s \in V$. This is known as the \emph{Single-Source Replacement Paths (SSRP)} Problem. For the SSRP problem, Chechik and Cohen~\cite{chechik2019near} and Gupta, Jain, and Modi~\cite{gupta2020multiple} presented $\widetilde{O}(m \sqrt{n}+n^2)$ time randomized algorithms for unweighted undirected graphs. Chechik, and Magen~\cite{ChechikM20} presented a randomized algorithm with the same runtime for unweighted directed graphs. For unweighted undirected graphs, Bil\`{o}, Cohen, Friedrich, and Schirneck~\cite{bilo2021near} designed the first deterministic algorithm with the same runtime. Later, Dey and Gupta \cite{DeyGupta22} improved the runtime to $\widetilde{O}(m \sqrt{n}+|R|)$, where $|R|$ is the output size. 

The \rpath{} problem in weighted directed graphs is known to belong to the fine-grained complexity classes of $\widetilde{\Theta}(n^3)$~\cite{williams2018subcubic} and $\widetilde{\Theta}(mn)$~\cite{AgarwalR18}.
Many of the problems in these complexity classes were known to have an $\widetilde{\Omega}(n)$ lower bound in the $\CONGEST$ model of distributed computing: APSP~\cite{nanongkai2014distributed}, diameter and radius~\cite{AbboudCK16,AnconaCDEW20}, and minimum weight cycle~\cite{manoharan2024computingMinWeightCycleCONGEST}.

In the area of fault-tolerant distributed computing, distributed algorithms for fault-tolerant distance preservers~\cite{BodwinP23, Parter20} and fault-tolerant spanners~\cite{DinitzR20, Parter22} have been shown in the $\CONGEST$ model of distributed computing.






\subsection{Paper Organization}
We begin with preliminaries in \Cref{sec:prelim}, followed by a technical overview of our results in \Cref{sec:overview}. Our \rpath{} algorithm for unweighted directed graphs is presented in \Cref{sec:short,sec:long}, where we consider replacement paths with short detours in \Cref{sec:short} and those with long detours in \Cref{sec:long}.
We establish our lower bound for unweighted directed graphs in \Cref{sec:lowerbound}.
We present our $(1+\epsilon)$-\apxrpath{} algorithm for weighted directed graphs in \Cref{sec:wapprox}.
 We conclude with open problems in \Cref{sec:conclusions}.

\section{Preliminaries}
\label{sec:prelim}

We start with introducing the basic notations used throughout the paper. 
\yanyu{set itemsep at the start. What do you mean by"(note: not allowed for camera-ready)"}
\yijun{While ok for arxiv, I think this is not allowed in the camera-ready version (breaking the given ACM style rule)}
\yanyu{check the 2,3,4 bullet}
\begin{itemize}
    \item $G = (V,E)$ is the input graph. 
    \item $s \in V$ and $t \in V$ are the \emph{source} and \emph{target}, respectively.
    \item $P = \stpath$
    is the given shortest path from $s$ to $t$ in $G$, where $h_{st}$ is the number of edges in $P$. 
    \item The \emph{length} of a path \( Q \), denoted as \( |Q| \), is the sum of the weights of the edges in \( Q \). If the graph is unweighted, $|Q|$ is simply the number of edges in $Q$. We emphasize that $h_{st}=|P|$ for unweighted graphs.
    \item $uv \diamond e$ is any shortest path from vertex $u$ to vertex $v$ avoiding edge $e$. More generally, for any path $R$, $uv \diamond R$ is any shortest path from vertex $u$ to vertex $v$ avoiding all the edges in $R$.
    \item $n = |V|$ is the number of vertices in the graph $G$.
    \item $m = |E|$ is the number of edges in the graph $G$.
    \item $D$ is the \emph{unweighted undirected} diameter of $G$.
    \item $\zeta = n^{2/3}$ is the threshold between a short detour and a long detour.
\end{itemize}
Given the notations, we define \rpath{}, \apxrpath{}, and \secsisp{} formally and succinctly.

\begin{definition}[Replacement Paths (\rpath{})]
For $0\leq i<h_{st}$, the \rpath{} problem requires each vertex $v_i$ in $P$ to compute $|st \diamond (v_i,v_{i+1})|$.
\end{definition}

\yanyu{mentioned input explicitly, should the id of s and t be given to all nodes? although this can be lifted by broadcasting.
In the previous paper, they assume shortest path Pst between the vertices s, t is part of the input and that each vertex in the network knows the identities of s and t, and the identities of vertices on Pst.
If we use this assumption, we can also omit the basic knowledge paragraph below.
However, now i realized that they need $h_st+D$ for preprocessing which we cannot afford.
}


\yanyu{other things to talk about the definition:
1: why this input is reasonable?
2: why requiring to output length; 
3: why output at the endpoint?
How are 2 and 3 useful in other context like fault-tolerant?
}

\begin{definition}[Approximate Replacement Paths (\apxrpath{})]
The goal of the $(1+\epsilon)$-\apxrpath{} problem is to let the first endpoint of each edge $e$ in $P$ compute a value $x$ such that $|st \diamond e| \leq x \leq (1+\epsilon)\cdot|st \diamond e|$.
\end{definition}

\begin{definition}[Second Simple Shortest Path (\secsisp{})]
The goal of the \secsisp{} problem is to let all vertices in $P$ compute the minimum value of $|st \diamond e|$ over all edges $e$ in $P$, which is the length of the second simple shortest path.
\end{definition}

\paragraph{Detour.} To search for a shortest replacement path for $e=(v_i, v_{i+1})$ (i.e., $st \diamond e$), it suffices to consider replacement paths of the form $(s, \ldots, v_j, \ldots, v_l, \ldots, t)$ with $j \leq i$ and $l \geq i+1$ such that both $(s, \ldots, v_j)$ and $(v_l, \ldots, t)$ are subpaths of $P$, and $(v_j, \ldots, v_l)$ does not share any edge with $P$.
We call $(v_j, \ldots, v_l)$ the \emph{detour} of the replacement path. If the detour has at most $\zeta$ \emph{hops}, then it is a \emph{short} detour. Otherwise, it is a \emph{long} detour. 

A replacement path with a short detour is called a short-detour replacement path. Similarly, a replacement path with a long detour is called a long-detour replacement path. For both \Cref{thm:main_UB} and \Cref{thm:apx_UB}, we set $\zeta = n^{2/3}$ to minimize the overall round complexity of our algorithms.  
The idea of considering replacement paths with short and long detours separately is from \citet{Roditty2012a}.

\paragraph{The $\CONGEST$ model.} In the $\CONGEST$ model of distributed computing, a network is represented by a graph \( G = (V, E) \), where each vertex corresponds to a computational device, and each edge corresponds to a communication link between. Each vertex $v$ has a unique $O(\log n)$-bit identifier \( \ID(v) \), initially known only to $v$ and its neighbors $N(v)$. The communication proceeds in \emph{synchronous} rounds. In each round, each vertex can send a (possibly distinct) message of  \( O(\log n) \) bits to each of its neighbors. If $G$ is a weighted graph, then we assume that the edge weights are positive integers in $\poly(n)$, ensuring that each edge weight fits within an $O(\log n)$-bit message. Each vertex is allowed to perform unlimited local computation.
the primary measure of complexity in the $\CONGEST$ model is the number of \emph{rounds} required for an algorithm to complete.
The following useful tool is well-known.


\begin{lemma}[Broadcast~\cite{peleg2000distributed}]\label{LP}
     Consider the routing task where each vertex $v$ wants to broadcast $m_v$ messages of $O(\log n)$ bits to all the vertices in the network. The task can be done in $O(M+D)$ rounds deterministically, where $M = \sum_{v \in V} m_v$ is the total number of messages.
\end{lemma}

\paragraph{Initial knowledge.} We now clarify the problem specification, specifically the input available to the vertices and, in particular, their knowledge of $(P, s, t)$. The algorithms presented in \Cref{thm:main_UB} and \Cref{thm:apx_UB} operate under the following minimal assumptions:
\begin{itemize}
    \item For each edge $e = \{u, v\}$ in $P$, both endpoints $u$ and $v$ initially know that $e$ belongs to $P$.
    \item The source vertex $s$ initially knows that it is the source.
    \item The target vertex $t$ initially knows that it is the target.
\end{itemize}
In the following lemma, we present a simple $\widetilde{O}(\sqrt{n} + D)$-round randomized algorithm that enables all vertices in $P$ to acquire additional useful information about the problem instance $(P, s, t)$.

\begin{lemma}[Obtaining additional information]\label{lem:basic_tool}
There exists an $\widetilde{O}(\sqrt{n}+D)$-round randomized algorithm that lets each vertex $v_i$ in $P$ learn its index $i$, the distance from $s$ to $v_i$, and the distance from $v_i$ to $t$, with high probability.
\end{lemma}
\begin{proof}
Before presenting the proof, observe that the distances can be computed using \emph{weighted undirected} SSSP by assigning a weight of $\infty$ to all edges outside $P$. Similarly, the indices can be computed using \emph{unweighted directed} SSSP by redirecting all edges incident to $P$ but outside $P$ towards vertices within $P$. In the following discussion, we present a simple $\widetilde{O}(\sqrt{n} + D)$-round randomized algorithm that achieves the same objective.

The algorithm begins by independently sampling each vertex in $P$ with probability $1/\sqrt{n}$. For every two consecutive sampled vertices $v_j$ and $v_l$ in $P$, our first task is to let every vertex $v$ along the $v_j$-$v_l$ subpath of $P$ learn the following information:
\begin{itemize}
    \item $\ID(v_j)$ and $\ID(v_l)$.
    \item The number of hops and the distances from $v_j$ to $v$ and from $v$ to $v_l$ in $P$.
\end{itemize}
Since all edge weights are integers bounded by some polynomial in $n$, this information can be encoded using $O(\log n)$ bits.

This learning task can be accomplished by performing a BFS from each sampled vertex, restricted to $P$, with each BFS terminating when it reaches a sampled vertex. This can be done in $O(\sqrt{n} \log n)$ rounds with high probability, as the number of hops between any two consecutive sampled vertices in $P$ is $O(\sqrt{n} \log n)$ with high probability, due to a Chernoff bound.

Subsequently, using the algorithm from \Cref{LP}, all sampled vertices broadcast their learned information, which consists of $O(\log n)$ bits. This requires ${O}(\sqrt{n}+D)$ rounds with high probability, as the total number of sampled vertices is $O(\sqrt{n})$ with high probability, again by a Chernoff bound. After the broadcast, each vertex $v_i$ has enough information to infer its index $i$, its distance from $s$, and its distance to $t$.
\end{proof}

Since the round complexity of \Cref{lem:basic_tool} lies within the target bounds of both \Cref{thm:main_UB} and \Cref{thm:apx_UB}, in their proofs, we may assume that each vertex $v_i$ in $P$ initially possesses the following information: 
\begin{itemize}
    \item Its index $i$.
    \item Its distance from the source $s$.
    \item Its distance to the target $t$.
\end{itemize}

\paragraph{Remark.} In the previous work by \citet{manoharan2024computing}, it is assumed that each vertex initially knows the identifiers of the vertices in $P$ in sequence. This assumption is justified in their setting, as the round complexity of their algorithms includes a term of $O(h_{st})$, which suffices for all vertices in $P$ to broadcast their identifiers. In contrast, since our goal is to eliminate the dependence on $h_{st}$, we cannot adopt this assumption.

\section{Technical Overview}
\label{sec:overview}

In this section, we provide a technical overview of our proofs. For each result, we first examine the approach used in prior work and then discuss how our method achieves an improvement. 
In \Cref{subsec:UB}, we discuss the proof idea behind \Cref{thm:main_UB}. 
In \Cref{subsec:LB}, we discuss the proof idea behind \Cref{thm:main_LB}. 
In \Cref{subsec:APX}, we discuss the proof idea behind \Cref{thm:apx_UB}.




\subsection{Upper Bound}\label{subsec:UB}
We start with the \rpath{} problem in unweighted directed graphs. Following an existing approach in centralized algorithms~\cite{Roditty2012a}, \citet{manoharan2024computing} approached the problem by separately handling replacement paths with short and long detours, using a threshold $\zeta$ to distinguish between them. For short-detour replacement paths, they compute a $\zeta$-hop BFS simultaneously from each vertex in $P$, which can be done in $O(h_{st} + \zeta)$ rounds~\cite{peleg2000distributed}.

To handle long-detour replacement paths, a set of \emph{landmark} vertices $L$ is sampled randomly. From each landmark vertex, an $O\left(\frac{n}{|L| \log n}\right)$-hop BFS is performed in parallel. The key idea is that the landmark vertices partition any long detour into segments of length $O\left(\frac{n}{|L| \log n}\right)$ with high probability. The distance information collected from the $O\left(\frac{n}{|L| \log n}\right)$-hop BFS explorations suffices to reconstruct all the shortest replacement path lengths involving a long detour. The idea of using landmark vertices to handle long paths in the \rpath{} problem originates from~\citet{BernsteinK08}, and can be traced back to earlier work on related graph problems~\cite{ullman1990high}.

To facilitate distance computation, \citet{manoharan2024computing} let each vertex in $L$ and $P$ broadcast all the distances computed during the BFS. This broadcasting step is the most time-consuming part of the algorithm, requiring $O(|L|^2 + |L|h_{st} + D)$ rounds: There are  $O(|L|+ h_{st})$ vertices in $L$ and $P$, and each of them needs to broadcast its distance from the $|L|$ landmark vertices, so the total number of messages in the broadcasting step is  $O(|L|^2 + |L|h_{st})$. Finally, by propagating the learned information along the path $P$, the optimal replacement path distances for all edges in $P$ can be obtained in $O(h_{st})$ rounds. By properly selecting the sampling probability for $L$ and the threshold $\zeta$, they obtained their round complexity upper bound $\widetilde{O}(n^{2/3} + \sqrt{n h_{st}} + D)$.

\paragraph{New approach.} Now, we present our improvement, which reduces the upper bound to $\widetilde{O}(n^{2/3} + D)$ and eliminates the dependence on $h_{st}$.



For short-detour replacement paths, we show that the shortest replacement path distances for all edges $e$ in $P$ can be computed deterministically in $O(\zeta)$ rounds. We still perform a $\zeta$-hop BFS from each vertex in $P$ in parallel but with a key modification that enables us to achieve this within $O(\zeta)$ rounds deterministically.  
A crucial observation is that it is unnecessary to fully expand all BFS trees. Instead, at each time step, we allow each vertex to propagate only the BFS tree originating from the \emph{furthest} vertex in $P$. This approach is sufficient for computing replacement paths while also resolving any potential congestion issues.  
After the BFS computation, we leverage the locality of short detours: The two endpoints of each short detour are at most $\zeta$ hops apart in $P$. This locality ensures that an additional $O(\zeta)$ rounds of information exchange within $P$ suffice to determine the optimal short-detour replacement paths.


Next, we consider the more challenging case of long-detour replacement paths. To minimize round complexity, we only let the landmark vertices broadcast the distances, and the vertices in $P$ do not broadcast, thereby reducing the broadcast cost to $O(|L|^2 + D)$ rounds. By setting $|L| = \widetilde{\Theta}(n^{1/3})$, we ensure that the overall round complexity remains within the target bound of $\widetilde{O}(n^{2/3} + D)$.  
However, this approach introduces some challenges. Unlike previous work, the distances between vertices in $P$ and those in $L$ are no longer globally known. In the previous work, the best replacement path for each edge $e$ was determined by propagating information sequentially from the start to the end of $P$, requiring $O(h_{st})$ rounds. For short-detour replacement paths, we leveraged locality to reduce this to $O(\zeta)$ rounds. However, this strategy does not apply to long-detour replacement paths, as a detour can start and end at two vertices in $P$ that are far apart.

To address these challenges, we partition the path \( P \) into \( \Theta(n^{1/3}) \) segments, each of length \( \Theta(n^{2/3}) \). For detours that begins and ends within the same segment, the optimal ones can be computed locally within the segment using \( O(n^{2/3}) \) rounds. For detours whose endpoints lie in different segments, we let each segment compute its best detour to every landmark vertex and broadcast the results. A key observation is that the total number of messages broadcast is only \( \tilde{O}(n^{2/3}) \), as we have $O(n^{1/3})$ segments and $|L| = \tilde{O}(n^{1/3})$ landmark vertices. Using these messages, along with the all-pairs distances within the set \( L \), the optimal replacement paths involving long detours can be efficiently computed.

\subsection{Lower Bound}\label{subsec:LB}

\citet{manoharan2024computing} established a lower bound of $\widetilde{\Omega}(\sqrt{n} + D)$ for the \rpath{} problem in unweighted directed graphs by reducing it from the undirected $s$-$t$ subgraph connectivity problem, which was known to admit an $\widetilde{\Omega}(\sqrt{n} + D)$ lower bound, as shown by the previous work~\cite{das2011distributed}. We briefly outline the framework of the lower bound proof in the work~\cite{das2011distributed}, which has been applied to derive $\widetilde{\Omega}(\sqrt{n} + D)$ lower bounds for a variety of distributed problems.

To obtain the lower bound, $\Theta(\sqrt{n})$ paths of length $\Theta(\sqrt{n})$ are created between two players, Alice and Bob, with one bit of information stored at the endpoint of each path on Bob's side. A constant-degree overlay tree of depth $O(\log n)$ on the set of all vertices is added to shrink the diameter to $O(\log n)$. The goal is to design the graph so that, to solve the considered problem correctly, Alice must learn all the bits stored on Bob's side. This can be achieved in two ways: The bits can either be transmitted along the paths or across the trees. In the first case, the \emph{dilation}, which is the maximum distance that some information must travel, is $\Omega(\sqrt{n})$. In the second case,  the \emph{congestion}, which is the maximum amount of information routed through any edge, is $\Omega(\sqrt{n})$.


\paragraph{New approach.} We extend this framework to establish a higher lower bound of $\widetilde{\Omega}(n^{2/3} + D)$. To achieve this bound, we construct $\Theta(n^{1/3})$ paths of length $\Theta(n^{2/3})$ between two players.  

On Bob's side, we create a complete bipartite graph on the set of $\Theta(n^{1/3})$ endpoints, storing $\Theta(n^{2/3})$ bits of information as the orientations of the edges in the bipartite graph. On Alice's side, we construct an $s$-$t$ path $P$ of length $\Theta(n^{2/3})$.  

The core idea of our construction is to establish a one-to-one correspondence between each edge $e$ in $P$ and each edge $e'$ in the complete bipartite graph, such that the shortest replacement path length for $e$ is determined by the orientation of the edge $e'$. To achieve this, we will design the lower bound graph carefully: For a given edge orientation, the shortest replacement path for $e$ will go through $e'$, but this is not possible if $e'$ is oriented in the opposite direction, resulting in a higher replacement path length.

While the idea of encoding quadratic information using a complete bipartite graph has been utilized in various lower bound proofs in the $\CONGEST$ model~\cite{frischknecht2012networks,manoharan2024computingMinWeightCycleCONGEST,manoharan2024computing}, our lower bound graph construction differs significantly from previous approaches. For instance, in the proof of an $\Omega(\sqrt{nk})$ lower bound for $k$-source shortest paths and distance sensitivity oracle queries by \citet{manoharan2024distributed}, the construction is symmetric on both sides, with bipartite graphs present in both Alice’s and Bob’s parts. In their construction, the horizontal paths connecting Alice and Bob are used \emph{unidirectionally}, with only one side of each bipartite graph linked to these paths.

In contrast, our construction is asymmetric and leverages horizontal paths in a \emph{bidirectional} manner, connecting both sides of the bipartite graph to these paths and forming a detour-shaped network topology. This key distinction allows us to establish the desired lower bound for the \rpath{} problem in unweighted directed graphs. 

\subsection{Approximation Algorithm}\label{subsec:APX} 

\citet{manoharan2024computing} demonstrated that their algorithm for the \rpath{} problem in \emph{unweighted directed} graphs can be adapted to solve the $(1 + \epsilon)$-\apxrpath{} problem in \emph{weighted directed} graphs, maintaining the same round complexity of $ \widetilde{O}(n^{2/3} + \sqrt{n h_{st}} + D)$. The primary modification required is to replace the $k$-hop BFS computation with the $(1 + \epsilon)$-approximation $k$-hop shortest paths. It was shown by Nanongkai~\cite[Theorem 3.6]{nanongkai2014distributed} that the $(1 + \epsilon)$-approximation $k$-hop shortest paths from $c$ sources can be computed in $ \widetilde{O}(k + c + D)$ rounds with high probability.


\paragraph{New approach.} 
While the transformation into an approximation algorithm is straightforward in prior work, this is not the case for us, as we aim to avoid any dependence on $h_{st}$. In particular, we cannot afford to perform $k$-hop shortest paths simultaneously from all vertices in $P$, as this would result in a round complexity of $\widetilde{O}(k + h_{st} + D)$, which has a linear dependence on $h_{st}$. Such a congestion issue seems inherent to weighted directed graphs: It is possible that the $k$-hop BFS from all vertices in $P$ overlap at the same edge. 

To overcome this issue, we use rounding to reduce the weighted case to the unweighted case. This allows us to apply our approach for running BFS in the short-detour case mentioned above, which does not suffer from any congestion issues, as each vertex only propagates one BFS in each step.

Additionally, we need a different approach to utilize the BFS distances for computing the replacement path length, as a short detour can start and end at two vertices that are arbitrarily far apart in $P$. The high-level idea is similar but different to how we handle long detours in the unweighted case. We break the path into segments and handle \emph{nearby} and \emph{distant} short detours separately. Detours that start or end within the same segment can be handled locally within the segment. For detours that cross two segments, the rough idea is to let each segment locally compute the best detour to every other segment and broadcast the result.


\section{Short-Detour Replacement Paths}
\label{sec:short}


In this section, we prove the following result. 


\begin{restatable}[Short detours]{proposition}{shortdetour}
\label{thm:shortdetour}
For unweighted directed graphs, there exists an $O(\zeta)$-round deterministic algorithm that lets the first endpoint $v_i$ of each edge $e=(v_i, v_{i+1})$ in $P$ compute the shortest replacement path length for $e$ with a short detour.
\end{restatable}

Recall that $\zeta$ is the threshold between a short detour and a long detour. While we set  $\zeta=n^{2/3}$ in \Cref{sec:prelim}, we emphasize that \Cref{thm:shortdetour} works for any choice of $\zeta$. \Cref{thm:shortdetour} is proved by an algorithm consisting of two stages, as follows. 

\begin{description}
    \item[Stage 1: Hop-constrained BFS.] We run a BFS from every vertex $v$ in $P$ for $\zeta$ hops \emph{backward}, meaning that the direction of the edges are reversed. To deal with any potential congestion issue, in each step, each vertex propagates only the BFS tree originating from the furthest vertex in $P$, dropping all the remaining branches, so all BFS can be completed in $\zeta$ rounds. This still allows each vertex in $P$ to learn the furthest vertex in $P$ that can be reached by a detour path starting from $v$ of length $d$, for all $d \in [\zeta]$. 
    \item[Stage 2: Information pipelining.] Using the information from hop-constrained BFS, each vertex can locally compute the best detour to skip all of the next $x$ edges in front of it, for all $x \in [\zeta]$. With a $(\zeta-1)$-round dynamic programming algorithm, the two endpoints of each edge $e$ in $P$ can compute the shortest replacement path length for $e$ with a short detour.
\end{description}

In this section, whenever $S = \emptyset$, we use the convention that $\max S = -\infty$ and $\min S = \infty$.

\subsection{Hop-Constrained BFS}

Recall that $P = \stpath$ is the given $s$-$t$ shortest path.
For each vertex $u$ in the graph $G$, we define $f_u^\ast(d)$ as the largest index $j$ satisfying the following conditions:
\begin{itemize}
    \item There exists a path from $u$ to $v_j$ of length exactly $d$ avoiding all edges in $P$.
    \item For any $\ell > j$, there is no path from $u$ to $v_\ell$  of length exactly $d$ avoiding all edges in $P$.
\end{itemize}
If such an index $j$ does not exist, we set $f_u^\ast(d) = -\infty$. 

\begin{restatable}[Hop-constrained BFS]{lemma}{backwardsBFS}
\label{Thm: Backward BFS Works}
There exists an $O(\zeta)$-round deterministic algorithm that lets each vertex $u$ in the graph $G$ compute $f_u^\ast(d)$ for all $d \in [\zeta]$.
\end{restatable}
\begin{proof}
Intuitively, it suffices to run a $\zeta$-hop  BFS from each vertex in $P$, where in each step, each vertex propagates only the BFS tree originating from the furthest vertex in $P$. We formalize this proof idea, as follows.

In the first round of the algorithm, each vertex $v_i$ in $P$ sends its index $i$ to all its incoming edges, excluding the edges in $P$. 
For each vertex $u$ in $G$, we write $S_d(u) \subseteq\{0,1,\ldots,h_{st}\}$ to denote the set of messages received by $u$ in the $d$th round of the algorithm. For the base case of $d = 0$, we set $S_0(v_i) = \{i\}$ for each vertex $v_i$ in $P$ and set $S_0(u) = \emptyset$ if $u$ does not belong to $P$.

For $d = 2, \ldots, \zeta$, in the $d$th round of the algorithm, for each vertex $u$ in $G$, if $S_{d-1}(u) \neq \emptyset$, then $u$ sends $\max S_{d-1}(u)$ to all its incoming edges, excluding the edges in $P$. Intuitively, $S_{d-1}(u)$ is the collection of BFS that reaches $u$ at round $d-1$, and the algorithm lets $u$ propagate only the BFS originating from the furthest vertex in $P$ in round $d$. 

It can be proved by an induction on $d$ that the following claim holds, so the algorithm indeed lets each vertex $u$ correctly compute $f_u^\ast(d)$ for each $d \in [\zeta]$.
\[f_u^\ast(d) = \max S_{d}(u).\]

The claim holds for $d=0$ by the definition of $S_0$. Assuming that the claim already holds for $d-1$, to see that the claim holds for $d$, observe that
\[f_u^\ast(d) = \max_{x \, : \, (u,x)\in E} f_{x}^\ast(d-1) = \max_{x \, : \, (u,x)\in E}\max S_{d-1}(x) = \max S_{d}(u),\]
as required.
\end{proof}

\subsection{Information Pipelining} 
Let us introduce some additional notations. For any $i < j$, we define $X[i, j]$ as the shortest length of a replacement path with a short detour that starts precisely at $v_i$ and ends precisely at $v_j$. If such a path does not exist, then $X[i, j] = \infty$. For example, we must have $X[i, j] = \infty$ whenever $j-i > \zeta$. 
Moreover, we define
\begin{align*}
X[\leq i,\,j] &= \min_{i' \, : \, i' \leq i} X[i', j],\\
X[i,\,\geq j] &= \min_{j' \, : \, j' \geq j} X[i, j'],\\
X[\leq i,\,\geq j] &= \min_{i' \, : \, i' \leq i} \, \, \min_{i' \, : \, j' \geq j} X[i', j'].
\end{align*}
Observe that $X[\leq i, \, \geq i+1]$ is precisely the shortest replacement path length with a short detour for the edge $e=(v_i, v_{i+1})$ in $P$, so our objective is to let each $v_i$ learn $X[\leq i, \, \geq i+1]$. The following lemma shows that the information obtained from BFS already allows each $v_i$ to calculate the value of $X[i, \, \geq j]$ for all $j > i$.

\begin{lemma}[Base case]\label{lem:short_base}
The value of $X[i, \, \geq j]$ for all $j > i$ can be determined from the value of $f_{v_i}^\ast(d)$ for all $d \in [\zeta]$.
\end{lemma}
\begin{proof}
We consider the following function:
\[
h^\ast(i,j) = \min\{ d\in[\zeta] \, | \, f_{v_i}^\ast(d)=j\}.
\]
We claim that
\[
X[i, \, \geq j]=\min\{X[i, \, \geq j+1], h_{st}-(j-i)+h^\ast(i,j)\}.\]
Observe that we must have $X[i, \, \geq j] = \infty$ whenever $j-i > \zeta$, so the above claim allows us to determine the value of $X[i, \, \geq j]$ for all $j > i$. For the rest of the proof, we prove the claim. 

\paragraph{Upper bound.} We prove the upper bound $X[i, \, \geq j] \leq \min\{X[i, \, \geq j+1], h_{st}-(j-i)+h^\ast(i,j)\}$. The part $X[i, \, \geq j] \leq X[i, \, \geq j+1]$ follows immediately from the definition of $X[i, \, \geq j]$. The part $X[i, \, \geq j] \leq h_{st}-(j-i)+h^\ast(i,j)$ follows from the fact that $h_{st}-(j-i)+h^\ast(i,j)$ equals the length of some replacement path that has a short detour from $v_i$ to $v_j$ of detour length $h^\ast(i,j)$.

\paragraph{Lower bound.} To show that  $X[i, \, \geq j] \geq \min\{X[i, \, \geq j+1], h_{st}-(j-i)+h^\ast(i,j)\}$, it suffices to show that, if $X[i, \, \geq j] < X[i, \, \geq j+1]$, then $X[i, \, \geq j] = h_{st}-(j-i)+h^\ast(i,j)$. The condition $X[i, \, \geq j] < X[i, \, \geq j+1]$ implies that $X[i, \, \geq j] = X[i, \, j]$, which equals the shortest length of a replacement path with a short detour that starts precisely at $v_i$ and ends precisely at $v_j$. If $h^\ast(i,j)$ equals the length $d$ of a shortest path from $v_i$ to $v_j$ avoiding the edges in $P$, then we have $X[i, \, \geq j] = X[i, \, j] = h_{st}-(j-i)+h^\ast(i,j)$, as required. Otherwise, by the definition of $f^\ast_{v_i}(d)$, there must be a vertex $v_{\ell}$ with $\ell > j$ such that there is a path from $u$ to $v_\ell$  of length exactly $d$ avoiding all edges in $P$, so $X[i, \, \geq j] \leq X[i, \ell] < X[i, \, j]$, which is a contradiction.
\end{proof}

\begin{lemma}[Dynamic programming]\label{lem:short_recursive}
Suppose each vertex $v_i$ in $P$ initially knows the value of $X[i, \, \geq j]$ for all $j > i$. There exists an $O(\zeta)$-round deterministic algorithm that lets each vertex $v_i$ in $P$ compute $X[\leq i, \, \geq i+1]$.
\end{lemma}
\begin{proof}
For the base case, each vertex $v_i$ can locally compute 
\[X[\leq i, \, \geq i+\zeta]=X[i, \, \geq i+\zeta],\]
since any short detour path that starts at or before $v_{i}$ and ends at or after  $v_{i+\zeta}$ must starts precisely at $v_{i}$ and ends precisely at $v_{i+\zeta}$.

For $d = \zeta, \zeta-1, \ldots, 2$, assuming that each vertex $v_i$ in $P$ already knows $X[\leq i, \, \geq i+d]$, then we can let each vertex $v_i$ in $P$ compute $X[\leq i, \, \geq i+(d-1)]$ using one round of communication. Observe that
\begin{align*}
  X[\leq i, \, \geq i+(d-1)] &= \min\{X[\leq i-1, \, \geq i+(d-1)], X[i, \, \geq i+(d-1)]\}\\
  &= \min\{X[\leq i-1, \, \geq (i-1) + d], X[i, \, \geq i+(d-1)]\},
\end{align*}
so it suffices to let each $v_{i-1}$ send $X[\leq i-1, \, \geq (i-1) + d]$ to $v_i$. Therefore, $\zeta-1$ rounds of communication suffice to let each vertex $v_i$ in $P$ compute $X[\leq i, \, \geq i+1]$.
\end{proof}
 
Combining all the ingredients, we prove the main result of the section.

\begin{proof}[Proof of \Cref{thm:shortdetour}]
 We first run the $O(\zeta)$-round algorithm of \Cref{Thm: Backward BFS Works} to let each vertex $u$ in the graph $G$ compute $f_u^\ast(d)$ for all $d \in [\zeta]$. By \Cref{lem:short_base}, the computed information allows each vertex $v_i$ in $P$ to infer the value of  $X[i, \, \geq j]$ any all $j > i$, enabling us to run the $O(\zeta)$-round algorithm of \Cref{lem:short_recursive}. By the end of the algorithm, each vertex $v_i$ in $P$ knows  $X[\leq i, \, \geq i+1]$, which is the shortest length of a replacement path for $e=(v_i, v_{i+1})$ with a short detour.
\end{proof}

\section{Long-Detour Replacement Paths}
\label{sec:long}

In this section, we handle the replacement paths having long detours, i.e., the length of the detour is greater than the threshold $\zeta = n^{2/3}$. We break the detours into smaller pieces so that the required distance information can be computed efficiently. To break the detours, we use the concept of \emph{landmark vertices}, which was previously used in the literature~\cite{BernsteinK08, ullman1990high} and have been widely used in fault-tolerant algorithms and distance computation. Moreover, as now a detour can cross two distant vertices in $P$, to facilitate information pipelining, we also divide the shortest $s$-$t$ path, $P$ into \emph{segments} separated by \emph{checkpoints}. At this point, let us define some notations which we will use throughout this section. Let $Q$ be any path.
\begin{itemize}
    \item We use the notation $u \leq_{Q} v$ to denote that $u$ appears before $v$ in the path $Q$.
    \item For any two vertices $x \leq_{Q} y$,  we write $Q[x,y]$ to denote the subpath of $Q$ from $x$ to $y$.
    \item Let $G \setminus Q$ denote the graph after deleting the edges in $Q$. 
    \item We write $(uv)_G$ to denote the shortest path from $u$ to $v$ in the graph $G$. We also write $|uv|_G$ to denote the length of the shortest path from $u$ to $v$ in the graph $G$. When the underlying graph is clear from the context, we omit the subscripts.
\end{itemize}

Let us now formally state the main result that we prove in this section. 

\begin{proposition}[Long detours]
\label{Thm: Long Detour Part Works}
    For unweighted directed graphs, there exists an $\widetilde{O}(n^{2/3}+D)$-round randomized algorithm that lets the first endpoint $v_i$ of each edge $e=(v_i, v_{i+1})$ in $P$ compute a number $x$ such that
    \[|st \diamond e| \leq x \leq \text{the shortest replacement path length for $e$ with a long detour}\] with high probability.   
\end{proposition}

\Cref{Thm: Long Detour Part Works} guarantees the exact value of $|st \diamond e|$ only when some shortest replacement path for $e$ takes a long detour. Otherwise, \Cref{Thm: Long Detour Part Works} provides only an upper bound on $|st \diamond e|$. In this case, the exact value of $|st \diamond e|$ can be computed using the algorithm from \Cref{thm:shortdetour}. Combining \Cref{thm:shortdetour} and \Cref{Thm: Long Detour Part Works}, we prove our first main theorem.

\mainUB*

\begin{proof}
Run the algorithms of \Cref{thm:shortdetour} and \Cref{Thm: Long Detour Part Works} with the threshold $\zeta = n^{2/3}$, by taking the minimum of the two outputs, the first endpoint $v_i$ of each edge $e=(v_i, v_{i+1})$ in $P$ correctly computes the shortest replacement path length $|st \diamond e|$.
\end{proof}




We start our discussion of the proof of \Cref{Thm: Long Detour Part Works} by defining the landmark vertices. 

\begin{definition}[Landmark vertices]
For a given constant $c  > 0$, sample each vertex in $V$ with probability $\frac{c \log n}{n^{2/3}}$ independently.
 The set $L$ of landmark vertices is the set of all sampled vertices.
\end{definition}

By a Chernoff bound, the size of the landmark vertex set is $|L| 
 = \widetilde{O}(n^{1/3})$ with high probability. The following proof is folklore and can be found in the textbook~\cite{greene1990mathematics}. For the sake of presentation, we write $L=\{l_1,l_2, \ldots, l_{|L|}\}$.

\begin{lemma}[Property of landmark vertices]
\label{Lem : Landmark Vertex}
    Consider any $S$ of $n^{2/3}$ vertices. With a probability of at least $1- n^{-\Omega(c)}$, i.e., with high probability, at least one of the vertices in $S$ is a landmark vertex.
\end{lemma}

Let us fix an edge $e = (v_i, v_{i+1})$ in $P$ whose replacement path $R = st \diamond e$ uses a long detour. Suppose we want $v_i$ to determine the length of $R$. Also, let the detour of $R$ begins at vertex $a$ and ends at vertex $b$. Since we are dealing with long detours, the subpath $R[a,b]$ contains more than $\zeta = n^{2/3}$ vertices. By \Cref{Lem : Landmark Vertex}, there exists at least one landmark vertex $l^* \in L$ on $R[a,b]$. To enable $v_i$ to compute $|R|$, it suffices to determine the lengths of the two subpaths of $R$: one from $s$ to $l^*$ and another from $l^*$ to $t$, denoted as $R[s, l^*]$ and $R[l^*, t]$, respectively. Observe that they can be rewritten as $
R[s, l^*] = sl^* \diamond P[v_{i+1}, t]$ and $R[l^*, t] = l^*t \diamond P[s, v_i]$. However, since the specific landmark vertices on the detour of $R$ are unknown, we must provide $v_i$ with the values of $|s l \diamond P[v_{i+1}, t]|$ and $|l t \diamond P[s, v_i]|$ for all $l \in L$. After that, $v_i$ can locally compute the length of $R$ as follows:

$$|st \diamond e| =\min\limits_{l \in L} |s l \diamond P[v_i,t]| + |lt \diamond P[s,v_{i+1}]|.$$

To summarize, for each vertex $v_i$ in the path $P = \stpath$, we want to store the following information in $v_i$. 
\begin{description}
    \item[1. Distances from $s$ to the landmark vertices:] $|sl \diamond P[v_{i},t]|$ for all $l \in L$.
    \item[2. Distances from the landmark vertices to $t$:] $|lt \diamond P[s,v_{i+1}]|$  for all $l \in L$.
\end{description}

In \Cref{s to land} we address the first part. In \Cref{land to t} we address the second part, which is symmetrical to the first part, and prove \Cref{Thm: Long Detour Part Works}. 

Observe that, for some $l \in L$, the distances $|sl \diamond P[v_{i},t]|$ and $|lt \diamond P[s,v_{i+1}]|$ may be more than $n^{2/3}$. So, a simple BFS-like algorithm may not work within $\widetilde{O}(n^{2/3}+D)$ rounds. To address that issue, we utilize a subroutine 
that enables each vertex $v \in V$ (hence $v \in P$) to know the distances between all pairs of landmark vertices.





\begin{lemma}[Distance between the landmark vertices]
\label{lem : land to land}
    There exists an $\widetilde{O}(n^{2/3}+D)$-round randomized algorithm that lets all vertices  $v \in V$ compute, for all pairs $(l_j,l_k) \in L \times L$, the length of the path $(l_jl_k)_{G \setminus P}$, with high probability.
\end{lemma}
Due to \Cref{Lem : Landmark Vertex}, the intuition is that BFS explorations up to depth $n^{2/3}$ starting from each vertex in $L$ should be sufficient to determine the distances between landmark vertices. We need the following lemma to formally prove \Cref{lem : land to land}.

\begin{lemma}[$k$-source $h$-hop BFS~\cite{LenzenPP19}]\label{kbfs}
    For directed graphs, there exists an $O(k+h)$-round deterministic algorithm that computes $k$-source $h$-hop BFS.
\end{lemma}

\begin{proof}[Proof of \Cref{lem : land to land}]
    We store $|l_jl_k|_{G \setminus P}$  in all $v \in V$, for any pair $(l_j,l_k) \in L \times L$, by finding the length of the whole path in parts. For that, observe that, using \Cref{Lem : Landmark Vertex}, in the path $(l_jl_k)_{G \setminus P}$, with high probability, there will be a landmark vertex in each length of $n^{2/3} $. That means, a longer $(l_jl_k)_{G \setminus P}$ path can be viewed as a concatenation of some paths between landmark vertices of at most $n^{2/3} $ hops. Therefore, if we can calculate the all-pair distances $L \times L$ in $G \setminus P$ under a constraint of $n^{2/3}$ hops, and broadcast the distances to all vertices $v \in V$, then $v$ can calculate $(l_jl_k)_{G \setminus P}$  for any pair of $|l_jl_k| \in L \times L$.

    In other words, we need a BFS tree of at most $h = n^{2/3}$ hops from all $l_i \in L$. We have in total $|L| = \widetilde{O}(n^{1/3})$ landmark vertices. Therefore, using \Cref{kbfs}, we can find the required BFS trees in $O(h+|L|) = O(n^{2/3})$ rounds. After that, each $l_i \in L$ knows the distances from all other landmark vertices under a hop constraint of $h = n^{2/3}$.


    Now, we plan to send this information to all the vertices $v \in V$. The total number of messages to be broadcast is $O(|L|^2) = \widetilde{O}(n^{2/3})$.    
    All the messages are of $O(\log n)$ bits. Therefore, by \Cref{LP}, the broadcast can be done in $\widetilde{O}(n^{2/3} + D)$ rounds.
    After that, each $v$ is able to locally calculate $|l_jl_k|_{G \setminus P}$  for all pairs $(l_j,l_k) \in L \times L$.

    While our algorithm is deterministic, the computed distances are only correct with high probability, as the guarantee of \Cref{Lem : Landmark Vertex} is probabilistic.
\end{proof}






\subsection{Distances From \texorpdfstring{$s$}{s} to the Landmark Vertices}
\label{s to land}

Let us recall that, the shortest path $P$ contains the vertices $\stpath$. Also, in each $v_i$, we want to store $|sl_j \diamond P[v_{i},t]|$ for all $l_j \in L$. Let us use the notation $M[l_j, v_i]$ to denote $|sl_j \diamond P[v_{i},t]|$, so our goal become storing $M[l_j,v_i]$, for all $l_j \in L$, in each $v_i$. 

At first, we informally discuss an approach to store $M[l_j,v_i]$, for all $l_j \in L$, in each $v_i$ in $\widetilde{O}(h_{st}+n^{2/3})$ rounds, and then show how to remove the additive term $O(h_{st})$.
Observe that, if you fix some $l_j$, then $M[l_j,v_i]$ is decreasing in $i$, i.e., $M[l_j,v_i] \geq M[l_j,v_{i'}]$ whenever $v_i \leq_{P}v_{i'}$. This is true because the path $sl_j \diamond P[v_{i},t]$  avoids $v_{i'}$ for all $i' \geq i$. 
Specifically, we have  \[M[l_j,v_i]= \min \left\{ M[l_j,v_{i-1}], |sv_i|+|v_il_j|_{G \setminus P} \right\}.\] Now, if we can store   $|v_il_j|_{G \setminus P}$, for all $l_j \in L$, in each $v_i$, then by doing $|L|$ sweeps in $P$ in a pieplining fashion, we can let each $v_i$ compute $M[l_j, v_i]$ for all $l_j \in L$ in $\widetilde{O}(h_{st}+|L|) = \widetilde{O}(h_{st}+n^{1/3})$ more rounds. We discuss how we can remove the $h_{st}$ additive term later. In view of the above, we want to let $v_i$ compute $|v_il_j|_{G \setminus P}$ for all $l_j \in L$. The next lemma achieves this goal in $\widetilde{O}(n^{2/3} + D)$ rounds.




\begin{lemma}[Distances to the landmark vertices] \label{lemma: vilj path}
    There exists an $\widetilde{O}(n^{2/3}+D)$-round randomized algorithm that lets each $v_i$ in $P$ compute $|v_il_j|_{G \setminus P}$, for all $l_j \in L$, with high probability.
\end{lemma}

\begin{proof}
First, let all vertices $v \in V$ know the length of the path $(l_j l_k)_{G \setminus P}$ for all pairs $(l_j, l_k) \in L \times L$. By \Cref{lem : land to land}, this can be achieved in $\widetilde{O}(n^{2/3} + D)$ rounds. 

Now, consider a vertex $v \in P$ and a landmark $l_j \in L$. Using \Cref{Lem : Landmark Vertex}, if the path $(v l_j)_{G \setminus P}$ has length at least $n^{2/3}$, then with high probability, there will be a landmark vertex in every segment of length $n^{2/3}$. This means that such a path $(v l_j)_{G \setminus P}$ can be viewed as a concatenation of the path from $v$ to some landmark $l$ and the path from $l$ to $l_j$, where $l$ is the first landmark encountered in the path from $v$ to $l_j$. Since, with high probability, the length of the path from $v$ to $l$ is at most $n^{2/3}$, and each vertex in $P$ already knows the distances between landmark vertices, it suffices to determine the distance of each vertex $v \in P$ to all vertices in $L$ under a hop constraint of $h=n^{2/3}$. 

To achieve this, we perform BFS from all landmark vertices $l_j \in L$ up to a depth of $n^{2/3}$ in $(G \setminus P)^R$ (the reverse graph of $G \setminus P$). Since $|L| = \widetilde{O}(n^{1/3})$, this requires performing $\widetilde{O}(n^{1/3})$ BFS computations, each up to $n^{2/3}$ hops. By \Cref{kbfs}, this can be completed in $O(|L| + n^{2/3}) = {O}(n^{2/3})$ rounds. During these BFS computations, each vertex can record its level in the BFS rooted at each $l_j$. At the end of the these BFS procedures, each vertex $v$ in $P$, knows the distances to all landmark vertices that can be reachable within $n^{2/3}$ hops.
\end{proof}

The length of pipelining $h_{st}$ is the bottleneck of the algorithm. Hence, to reduce the number of rounds, we want to reduce the required length of pipelining. For that, we break the path $P$ into segments of length $O(n^{2/3})$. We call the endpoints of the segments as checkpoints. Formally, 
let the set of checkpoints be \[C=\{s=v_0,v_{\lceil n^{2/3} \rceil},v_{2\lceil n^{2/3} \rceil},,v_{3\lceil n^{2/3} \rceil} \ldots, t\}.\] 
Therefore, between two consecutive checkpoints, there are exactly $\lceil n^{2/3} \rceil$ edges, except for the last segment, where the number of edges can be smaller, as we put $t$ as the last checkpoint.
We denote the checkpoint as $c_1, c_2 \ldots$. The segment $P[c_{i},c_{i+1}]$, which starts from the vertex $c_{i}$ and ends at the vertex $c_{i+1}$ in $P$, is the $i$th segment. The segments are edge-disjoint and overlap in the checkpoints. 

For our next lemma, let us define a ``localized'' version of $M[l_j,v_i]$ restricted to the $i$th segment. For each vertex $v$ in $P$ such that $c_{i} \leq_P v \leq_P c_{i+1}$, we define \[M^i[l_j,v]=\min_{u \, : \, c_{i} \leq_P u \leq_P  v} \left\{|su|+|ul_j|_{G \setminus P}\right\}.\]



%


\begin{lemma}[Information pipelining within a segment] \label{mincheckpoint}
    There exists an $\widetilde{O}(n^{2/3}+D)$ randomized algorithm that lets each vertex $v$ in each segment $P[c_i, c_{i+1}]$  compute $M^i[l_j,v]$ with high probability. 
        \end{lemma}

\begin{proof}
First, let all vertices $v \in P$ compute  $|vl_j|_{G \setminus P}$ for all $l_j \in L$. By \Cref{lemma: vilj path}, this can be achieved in $\widetilde{O}(n^{2/3}+D)$ rounds.


Let us focus on the $i$th segment, $P[c_{i},c_{i+1}]$.
Let us now fix $l_j$ and discuss how we let all the vertices $v$ in the $i$th segment $P[c_{i},c_{i+1}]$ compute the value of $M^i[l_j,v]$.  In the first round, the first vertex $c_{i}$ calculates $M^i[l_j,c_{i}] = |sc_i| + |c_i l_j|_{G \setminus P}$ locally and sends the value to the second vertex $v'$ of the segment. In the second round, the second vertex calculates $M^i[l_j,v']=\min\{M^i[l_j,c_{i}], |sv'| + |v' l_j|_{G \setminus P}\}$ and send the value to the third vertex of the segment, and so on. The number of rounds needed equals the number of edges in the segment, which is $O(n^{2/3})$.

Observe that, in the $x$th round, we only send a message in the $x$th edge of the segment. 
Therefore, by pipelining (i.e., the algorithm for $l_j$ is delayed by $j$ rounds), we can run the algorithms for all $l_j \in L$ in parallel using $O(n^{2/3}) + |L|-1 = O(n^{2/3})$ rounds.
\end{proof}


Consider the vertex common to both the $(i-1)$th and $i$th segments, i.e., the checkpoint $c_i$. For a fixed landmark vertex $l_j$, both $M^{i-1}[l_j, c_i]$ and $M^{i}[l_j, c_i]$ are well-defined. However, these values may differ. Both $M^{i-1}[l_j, c_i]$ and $M^{i}[l_j, c_i]$ are necessary for our proof in the following lemma. 

\begin{lemma}[Part 1]
\label{lem : s to land}
There exists an $\widetilde{O}(n^{2/3}+D)$-round randomized algorithm that lets each vertex $v_i$ in $P$ compute the length of the path $sl_j \diamond P[v_{i},t]$, for all $l_j \in L$, with high probability.
\end{lemma}

\begin{proof}
We first run the algorithm of \Cref{mincheckpoint}. After that, using \Cref{LP}, we can broadcast all values of $M^{i}[l_j,c_{i+1}]$, for all $l_j \in L$ and all segments $P[c_i, c_{i+1}]$, to all the vertices in $\widetilde{O}(n^{2/3}+D)$ rounds. Now, we claim that we already have all the necessary ingredients to calculate \[|sl_j \diamond P[v_{i},t]|= \min_{u \, : \, s \leq_P u \leq_P v_i} \{|su|+|ul_j|_{G \setminus P}\}\] for each $l_j \in L$ at $v_i$. Assume that $(v_i,v_{i+1})$ lies in the segment between $c_k$ and $c_{k+1}$. There are two cases, depending on the vertex $u$ where the detour of $sl_j \diamond P[v_{i},t]$ starts.

\begin{enumerate}

    \item If the detour starts at the same segment, then the length of $sl_j \diamond P[v_{i},t]$ is 
    \[M^k[l_j,u] = \min_{u \, : \, c_k \leq u \leq v_i} \left\{|su|+|ul_j|_{G \setminus P}\right\}.\] The vertex $v_i$ has learned this value from the algorithm of \Cref{mincheckpoint}.

    \item If the detour starts at some other segment, then the length of $sl_j \diamond P[v_{i},t]$ is \[ \min_{c_x \in C \, : \, x < k} \, \,  \, \min_{u \, : \, c_k \leq u \leq v_{k+1}}  \left\{|su|+|ul_j|_{G \setminus P}\right\} = \min_{c_x \in C \, : \, x < k} M^{x}[l_j,c_{k+1}].\]  The vertex $v_i$ as learned  all the values of $M^{x}[l_j,c_i]$ from the broadcast.
\end{enumerate}
The length of $sl_j \diamond P[v_{i},t]$ is the minimum value of the above two cases.
\end{proof}

\subsection{Distances From the Landmark Vertices to \texorpdfstring{$t$}{t}}\label{land to t}

This third part is be symmetric to \Cref{s to land}. 


\begin{lemma}[Part 2]
\label{lem : land to t}
    There exists an $\widetilde{O}(n^{2/3}+D)$-round randomized algorithm that lets each vertex $v_i$ in $P$ compute the length of the path $l_jt \diamond P[s,v_{i+1}]$, for all $l_j \in L$, with high probability.
\end{lemma}

\begin{proof}
    To prove this lemma, we consider the reverse graph $G^R$ of $G$ by reversing all the edges of $G$. The path $P$ will also be reversed to $P^R$, which is $(t=v'_0,v'_1, \dots , v'_{h_{st}}=s)$, where the vertex $v_i$ in the $P$ is represented as $v'_{h_{st}-i}$ in $P^R$. Our goal is to let $v_i$ compute the distances $|l_jt \diamond P[s,v_{i+1}]|$ for all $l_j \in L$ in the graph $G$. This is equivalent to letting each $v'_i$ compute the distances $|tl_j \diamond (v'_{i-1},s)|$  for all $l_j \in L$ in the graph $G^R$, which can be done using the algorithm of \Cref{lem : s to land} in  $\widetilde{O}(n^{2/3}+D)$ rounds. 
    
    A slight mismatch here is that the algorithm of \Cref{lem : s to land} will store the distance $|l_jt \diamond P[s,v_{i+1}]|$ in $v_{i+1}$ and not $v_i$. This can be resolved by transmitting the information from $v_{i+1}$ to $v_i$ along the edge $(v_i, v_{i+1})$, which costs additional $O(|L|)=\widetilde{O}(n^{1/3})$ rounds.
\end{proof}

We are now ready to prove \Cref{Thm: Long Detour Part Works}.

\begin{proof}[Proof of \Cref{Thm: Long Detour Part Works}]

    

The algorithm runs the procedures of \Cref{lem : s to land,lem : land to t} by spending  $\widetilde{O}(n^{2/3}+D)$ rounds. From the information stored in $v_i$, the vertex $v_i$ can locally compute 
     $$x=\min _{l \in L} \{ |sl \diamond P[v_i,t]| + |lt \diamond P[s,v_{i+1}]\}.$$
Observe that  $x$ is the minimum of the lengths of all paths from $s$ to $t$ that avoids $e$ and contains at least one landmark vertex in $L$.     
    
    We show that the output $x$ satisfies the requirement. Let the path that realizes value $x$ be $P_i$.
     As $P_i$ is a path from $s$ to $t$ avoiding $e$, we have $x \geq |st \diamond e|$.  Let $P_i^*$ be a shortest path from $s$ to $t$ avoiding $e$ with a long detour. By \Cref{Lem : Landmark Vertex}, $P_i^*$ contains a landmark vertex with high probability. Hence, by the definition of $x$, we have $x \leq |P_i^*|$, as required.
\end{proof}


\newcommand{\GGdp}{G(\Gamma, d, p)}
\newcommand{\Gkdpp}{G(k, d, p, \phi)}
\newcommand{\GkdppMx}{G(k, d, p, \phi, M, x)}
\newcommand{\GGdptitle}{\texorpdfstring{$\GGdp$}{G(Gamma, d, p)}}
\newcommand{\Gkdpptitle}{\texorpdfstring{$\Gkdpp$}{G(k, d, p, phi)}}
\newcommand{\bin}{\{0,1\}}
\newcommand{\disj}{\mathsf{disj}}
\newcommand{\inprod}[2]{\langle #1, #2\rangle}
\newcommand{\Pstar}{\mathcal{P^*}}
\newcommand{\Pc}{\mathcal{P}}
\newcommand{\Qc}{\mathcal{Q}}
\newcommand{\Rc}{\mathcal{R}}
\newcommand{\Ac}{\mathcal{A}}
\newcommand{\Tc}{\mathcal{T}}
\newcommand{\Mc}{\mathcal{M}}

\section{Lower Bound}
\label{sec:lowerbound}


In this section, we present our \emph{randomized} $\widetilde{\Omega}(n^{2/3} + D)$ lower bound for the \emph{exact} versions of both \secsisp{} and \rpath{}. A randomized algorithm that fails with a probability of at most $\epsilon$ is called an \emph{$\epsilon$-error} algorithm. We consider the more general $\CONGEST(B)$ model where each vertex is allowed to send a $B$-bit message to each of its neighbors in each round. The standard $\CONGEST$ model is a special case of the $\CONGEST(B)$ model with $B= \Theta(\log n)$.

Our main results are stated as follows.

\begin{restatable}[\secsisp{} lower bound]{proposition}{rpathlower}
\label{thm:2sisp lower}
    For any $p\geq 1$, $B\geq 1$ and $n\in \left\{d^{3p/2}\mid d\geq 2 \right\}$, there exists a constant $\epsilon>0$ such that any $\epsilon$-error distributed algorithm for the \secsisp{} problem requires $\Omega\left(\frac{n^{2/3}}{B\log n}\right)$ rounds on some $\Theta(n)$-vertex unweighted directed graph of diameter $2p+2$ in the $\CONGEST(B)$ model.
\end{restatable}

With the lower bound for the \secsisp{} problem in \Cref{thm:2sisp lower}, we can derive a corollary that provides a lower bound for the \rpath{} problem. This follows from the fact that the \secsisp{} problem can be reduced to the replacement path problem, requiring only additional $O(D)$ rounds to compute the minimum replacement path for each edge along the given $s$-$t$ path.

\begin{corollary}[\rpath{} lower bound]\label{thm:rpath lower}
    For any $p\geq 1$, $B\geq 1$ and $n\in \left\{d^{3p/2}\mid d\geq 2 \right\}$, there exists a constant $\epsilon>0$ such that any $\epsilon$-error distributed algorithm for the \rpath{} problem requires $\Omega\left(\frac{n^{2/3}}{B\log n} \right)$ rounds on some $\Theta(n)$-vertex unweighted directed graph of diameter $2p+2$ in the $\CONGEST(B)$ model.
\end{corollary}

In \Cref{thm:2sisp lower} and \Cref{thm:rpath lower}, we demonstrate an infinite sequence of graphs with diameter $D = \Theta(\log n)$ on which the \secsisp{} and \rpath{} problems require $\widetilde{\Omega}(n^{2/3})$ rounds to solve. We extend the $\widetilde{\Omega}(n^{2/3})$  lower bound to an $\widetilde{\Omega}(n^{2/3}+D)$ lower bound, as follows.

\mainLB*

\begin{proof}
As \Cref{thm:2sisp lower} and \Cref{thm:rpath lower} already establish an $\widetilde{\Omega}(n^{2/3})$ lower bound, it remains to show an $\Omega(D)$ lower bound. To do so, for any integer $D$, we construct a graph of diameter $D$ on which the \secsisp{} and \rpath{} problems require $\Omega(D)$ rounds to solve. The construction consists of two parallel directed paths from $s$ to $t$: one of length $D$ and the other of length $D+1$. We allow zero or one edge in the longer path to be reversed. Consequently, the length of the second shortest path from $s$ to $t$ is either $D+1$ or $\infty$, and distinguishing between these two cases requires $\Omega(D)$ rounds. This graph has $n = 2D + 1$ vertices. To extend the lower bound to any $n \geq 2D + 1$, we can attach a clique of size $n - (2D + 1)$ to any edge in the graph.
\end{proof}
\yanyu{i change vertex to edge, since this way the diameter will still be D.}
\yijun{slightly edited it}
    



\paragraph{Technical framework.} 
Our lower bound follows the framework of the work~\cite{das2011distributed}, which uses the graph $\GGdp$ described later. They showed how to simulate distributed algorithms for the computing functions in certain networks in the communication complexity model, thereby establishing lower bounds for computing functions like disjointness in those networks. They further reduced computing functions in these networks to solving the targeted problem. In their paper~\cite{das2011distributed}, they showed how various verification problems,  such as connectivity verification,  can be used to compute disjointness in $\GGdp$, hence demonstrating an $\widetilde{\Omega}(\sqrt{n})$ lower bound for these problems.


\paragraph{Roadmap.} 
In \Cref{subsec:comm comp}, we first review some terminologies and basic settings. 
Next, in \Cref{subsec:ggdp and sim}, we describe the graph $\GGdp$ used in the work~\cite{das2011distributed}, and restate their simulation lemma that enable us to link communication complexity to distributed computation of functions. 
Then, in \Cref{subsec:modification}, we describe our construction $\Gkdpp$ and its directed version $\GkdppMx$ which are modified from $\GGdp$ and show that any algorithm on the modified graph $\Gkdpp$ can be simulated on $\GGdp$ with $O(1)$-factor overhead. Lastly, in \Cref{subsec:lower bound reduction}, we conclude by showing that computing the disjointness function in $\Gkdpp$ can be reduced to computing the second simple shortest path (\secsisp{}) and thus showing a lower bound on both the \secsisp{} problem and the \rpath{} problem in the $\CONGEST$ model. 




\subsection{Communication Complexity}
\label{subsec:comm comp}

We review the terminology from the work~\cite{das2011distributed} regarding the complexity of computing a Boolean function $f: \bin^b \times \bin^b \to \bin$ in both the two-party communication model and the $\CONGEST(B)$ model.

\begin{description}
    \item[Communication complexity:] Suppose Alice is given $x\in \bin^b$ and Bob is given $y\in \bin^b$. The goal is to send messages to each other for multiple rounds to correctly output the value of $f(x,y)$ at the end of the process. For any $\epsilon>0$ and any Boolean function $f$, let $R_\epsilon^{\operatorname{cc-pub}}(f)$ denote the minimum worse-case communication in bits for both Alice and Bob to output $f(x,y)$ correctly with at most $\epsilon$ failure with \emph{public randomness}.
    \item[Distributed computation of $f$ in $G$:] Given a graph $G = (V,E)$, and two distinguished vertices $\alpha, \beta\in V$. Suppose Alice received $x\in \bin^b$ as input on vertex $\alpha$ and Bob received $y\in \bin^b$ as input on vertex $\beta$. They are to communicate via the network $G$ under the $\CONGEST(B)$ model constraint and output $f(x,y)$ at the end of the process. For any $\epsilon>0$, any graph $G$ and any Boolean function $f$, define $R_\epsilon^G(f; B)$ as the minimum worst-case number of rounds of the best $\epsilon$-error randomized distributed algorithm for computing $f$ on $G$ under the $\CONGEST(B)$ model. We often write $R_\epsilon^G(f)$ when $B$ is clear in the context. Observe that the first model can be seen as a special case of this model with $B=1$ and $G=(\{\alpha,\beta\}, \{\{\alpha,\beta\}\})$ being an edge.
\end{description}

In this paper, we only consider the set disjointness Boolean function defined below, where $\inprod{x}{y}\coloneqq \sum_{i=1}^b x_i y_i$ denotes the inner product between vectors $x$ and $y$.
$$
\disj_b: \bin^b \times \bin^b \to \bin, \; \text{ where } \;  \disj_b(x,y)=\begin{cases}
    1 & \text{if }\inprod{x}{y} =0;\\
    0 & \text{otherwise}.
\end{cases}
$$

\subsection{The Graph \GGdptitle{} and the Simulation Lemma}
\label{subsec:ggdp and sim}
First, we describe the graph family $\GGdp$ parameterized by $\Gamma, d$ and $p$. In short, $\GGdp$ consists of $\Gamma$ copies of $d^p$-vertex paths, all connected to the leaves of a $p$-depth, $d$-branch tree. Formally, we name the $i$th path by $\mathcal{P}^i$ and we name the vertices in $\mathcal{P}^i$ by $v^i_0, v^i_1, \dots, v^i_{d^p-1}$ in order. We index the vertices in the tree $\Tc$ by their depth and their order in the level. The root of $\Tc$ is $u_0^0$, and the leaves are $u_0^p, u_1^p, \dots,u_{d^p-1}^p$. Besides the path and tree edges, we further add edges $\{u_i^p, v_i^j\}$ for all $0\leq i\leq d^p-1$ and $1\leq j\leq \Gamma$. Lastly, we set $\alpha=u^p_0$ and $ \beta = u^p_{d^p-1}$.
Refer to \Cref{fig:GGdp} for an illustration. We have the following observation.

\begin{observation}[Basic properties of $\GGdp$~\cite{das2011distributed}]
    There are $\Gamma d^p + \frac{d^{p+1}-1}{d-1}=\Theta(\Gamma d^p)$ vertices in $\GGdp$, and the diameter of $\GGdp$ is $2p+2$.
\end{observation}

\begin{figure}[htbp]
    \centering
    \begin{tikzpicture}[scale=1.1]
    \tikzset{
        vertex/.style={circle, draw, minimum size=4pt, inner sep=1pt, fill=white},
        filled/.style={circle, draw, fill=black, minimum size=4pt, inner sep=1pt},
    }

    \node[] at (-1,0) {$\mathcal{P}^1$};
    \node[vertex] (a) at (0,0) [label=below:$v^1_0$] {};
    \node[vertex] (b) at (1,0) [label=below:$v^1_1$] {};
    \node[vertex] (c) at (2,0) [label=below:$v^1_2$] {};
    \node[] (ca) at (3,0) {};
    \node[] at (3.5,0) {$\dots$};
    \node[] (cb) at (4,0) {};
    \node[vertex] (d) at (5,0) [label=below:$v^1_{d^p-1}$] {};

    \node (dda) at (0,-1) {$\vdots$};
    \node (ddb) at (1,-1) {$\vdots$};
    \node (ddc) at (2,-1) {$\vdots$};
    \node (ddd) at (5,-1) {$\vdots$};
    \node (ddda) at (0,-2.5) {$\vdots$};
    \node (dddb) at (1,-2.5) {$\vdots$};
    \node (dddc) at (2,-2.5) {$\vdots$};
    \node (dddd) at (5,-2.5) {$\vdots$};
    
    \node[] at (-1,-1.5) {$\mathcal{P}^\ell$};
    \node[vertex] (e) at (0,-1.5) [label=below:$v^\ell_0$] {};
    \node[vertex] (f) at (1,-1.5) [label=below:$v^\ell_1$] {};
    \node[vertex] (g) at (2,-1.5) [label=below:$v^\ell_2$] {};
    \node[] (ga) at (3,-1.5) {};
    \node[] at (3.5,-1.5) {$\dots$};
    \node[] (gb) at (4,-1.5) {};
    \node[vertex] (h) at (5,-1.5) [label=below:$v^\ell_{d^p-1}$] {};
    
    \node[] at (-1,-3.1) {$\mathcal{P}^\Gamma$};
    \node[vertex] (i) at (0,-3.1) [label=below:$v^\Gamma_0$] {};
    \node[vertex] (j) at (1,-3.1) [label=below:$v^\Gamma_1$] {};
    \node[vertex] (k) at (2,-3.1) [label=below:$v^\Gamma_2$] {};
    \node[] (ka) at (3,-3.1) {};
    \node[] at (3.5,-3.1) {$\dots$};
    \node[] (kb) at (4,-3.1) {};
    \node[vertex] (l) at (5,-3.1) [label=below:$v^\Gamma_{d^p-1}$] {};
    
    \node at (5.5,3) {$\mathcal{T}$};
    \node[vertex] (o) at (3.6,3.5) [label=right:$u_0^0$] {};
    
    \node (left) at (2.7,3.05) {};
    \node (right) at (4.45,2.8) {};
    
    \node (eleft) at (2.45,3) {$\iddots$};
    \node (eright) at (4.8,2.6) {$\ddots$};
    \node[vertex] (pq) at (2,2.6) {};
    
    \node[vertex] (p) at (1,1.85) [label=right:$u^{p-1}_0$] {};
    \node[vertex] (q) at (3,1.85) [label=right:$u^{p-1}_1$] {};
    \node[vertex] (qq) at (5.4,1.85) {};
    
    \node[filled] (s) at (0.5,1) [label=left:{$\alpha=u^p_0$}] {};
    \node[vertex] (m) at (1.5,1) [label=left:$u^p_1$] {};
    \node[vertex] (n) at (2.5,1) [label=left:$u^p_2$] {};
    \node[vertex] (nn) at (3.5,1) {};
    \node (mid) at (4.5,1) {$\cdots$};
    \node[vertex] (rl) at (5,1) {};
    \node[filled] (r) at (5.8,1) [label=right:{$\beta=u^p_{d^p-1}$}] {};
    
    \draw (a) -- (b) -- (c) -- (ca);
    \draw (e) -- (f) -- (g) -- (ga);
    \draw (i) -- (j) -- (k) -- (ka);
    \draw (cb) -- (d);
    \draw (gb) -- (h);
    \draw (kb) -- (l);
    
    \draw (o) -- (left);
    \draw (o) -- (right);
    \draw (p) -- (pq);
    \draw (q) -- (pq);
    \draw (p) -- (s);
    \draw (p) -- (m);
    \draw (q) -- (n);
    \draw (q) -- (nn);
    \draw (qq) -- (r);
    \draw (qq) -- (rl);
    
    \draw (s) to [out=265, in=35] (a);
    \draw (s) to [out=268, in=65] (e);
    \draw (s) to [out=273, in=70] (i);
    \draw (m) to [out=265, in=35] (b);
    \draw (m) to [out=268, in=65] (f);
    \draw (m) to [out=273, in=70] (j);
    \draw (n) to [out=265, in=35] (c);
    \draw (n) to [out=268, in=65] (g);
    \draw (n) to [out=273, in=70] (k);
    \draw (r) to [out=245, in=20] (d);
    \draw (r) to [out=260, in=40] (h);
    \draw (r) to [out=270, in=55] (l);
\end{tikzpicture}
    \caption{An example of $\GGdp$ with $d=2$.}
    \label{fig:GGdp}
\end{figure}

We restate the simulation lemma of the work~\cite{das2011distributed}, which connects the computation of Boolean function in the traditional communication complexity model to the computation of Boolean function in the distributed model in graph $\GGdp$ with input vertices $\alpha$ and $\beta$.

\begin{lemma}[Simulation lemma \cite{das2011distributed}]
\label{thm:simulation}
For any $\Gamma$, $d$, $p$, $B$, $\epsilon\geq 0$, and function $f:\bin^b \times \bin^b \to \bin$, if there is an $\epsilon$-error randomized distributed algorithm that computes $f(x,y)$ faster than $\frac{d^p-1}{2}$ rounds on $\GGdp$ with $x$ given to vertex $\alpha$ and $y$ given to vertex $\beta$, i.e.,
$
R_\epsilon^{\GGdp}(f) < \frac{d^p-1}{2},
$
then there is an $\epsilon$-error randomized algorithm in the communication complexity model that computes $f$ using at most $2dpB R_\epsilon^{\GGdp}(f)$ bits of communication.
In other words,
$
R_\epsilon^{\operatorname{cc-pub}}(f)\leq 2dpB R_\epsilon^{\GGdp}(f).
$
\end{lemma}


The simulation lemma was proved by analyzing the states of the distributed (sub)-network over rounds. In $\GGdp$, Alice and Bob have two ways to communicate $b$ bits: either through paths of length $d^p-1$ between them, or through the tree structure. When an algorithm's runtime is less than $(d^p-1)/2$, messages cannot traverse the full path between Alice and Bob, forcing communication through the tree structure. This creates a congestion bottleneck which is manifested by their analysis that only $O(dp)$ messages are needed to simulate all messages sent in each round of the distributed algorithm on $\GGdp$. This allows Alice and Bob to simulate any distributed algorithm for computing functions $f$ on $\alpha$ and $\beta$ on $\GGdp$ in the communication complexity model with an $O(dp)$-factor overhead.

The main way the above simulation lemma is used is via the following lemma which applies the simulation lemma specifically for the $\disj$ function. It translates the lower bound for $\disj$ in the communication model to a lower bound for computing $\disj$ in $\GGdp$ on $\alpha$ and $\beta$. 

\begin{lemma}[Set disjointness lower bound for $\GGdp$~\cite{das2011distributed}]\label{lem:disj}
For any $\Gamma, d, p$, there exists a constant $\epsilon>0$ such that
$$R_\epsilon^{\GGdp}\left(\disj_b;B\right)=\Omega\left(\min\left(d^p, \frac{b}{dpB}\right)\right).$$
\end{lemma}

\subsection{Modification to \GGdptitle{}}
\label{subsec:modification}

In this section, we present our graph construction $\Gkdpp$ builds upon $\GGdp$ to allow us to perform a reduction from distributed set disjointness to the \secsisp{} and \rpath{} problems. Compared with the $\widetilde{\Omega}(\sqrt{n})$ lower bounds obtained via existing constructions~\cite{das2011distributed,manoharan2024computing}, our graph construction $\Gkdpp$ allows us to obtain a higher $\widetilde{\Omega}(n^{2/3})$ lower bound. 
We show how to simulate any distributed disjointness algorithm for $\Gkdpp$ on $\GGdp$ with $O(1)$-factor overhead. Connecting these two results, we obtain a lower bound for computing disjointness on $\Gkdpp$ with inputs given to $\alpha$ and $\beta$.

\paragraph{Intuition.} Our lower bound construction is built around a bipartite graph positioned at the far end of the structure, which controls the replacement path distances for each edge along the given $s$-$t$ path. We aim to establish a correspondence between the edge orientations in the bipartite graph and the replacement path distances for the edges in the given $s$-$t$ path. Moreover, we require that reversing the orientation of the edges in the bipartite graph result in longer replacement paths. To solve the replacement path problem, Alice must learn the orientation of every edge in the bipartite graph, requiring substantial information to be transmitted from one end to the other.

The construction proceeds by first establishing a mapping between each edge on the $s$-$t$ path and its corresponding edge in the bipartite graph. 
For each edge in the $s$-$t$ path, we connect a long path to its mapped edge in the bipartite graph.
The long distance between the output vertices (vertices in the given $s$-$t$ path) and the location of critical information (edge orientations in the bipartite graph) forces the algorithm to spend more rounds propagating information to their required destinations. 
The orientation of the edges in the bipartite graph determines the length of the replacement path for their corresponding edges. 
With $n^{1/3}$ paths of length approximately $n^{2/3}$, the algorithm requires at least $n^{2/3}$ rounds to propagate this critical information. 
Next, to reduce the graph's diameter without creating additional replacement paths, we incorporate a tree-like structure with downward edge orientations along the path. As demonstrated in the simulation lemma of the work~\cite{das2011distributed} and our new simulation in \Cref{lem:disj G(kdpp)}, this structure does not substantially improve the propagation of critical information.
Finally, to optimize the total number of vertices, we carefully merge paths that lead to the same vertex in the bipartite graph while ensuring that the correspondence between edge orientations in the bipartite graph and replacement path lengths remains intact. 
This is done by leaving enough distance between consecutive vertices that goes to the same vertex in the bipartite graph.
\yanyu{added last sentence, cuz comment asks to say more.}

\paragraph{Construction of \Gkdpptitle.}
Our modified graph is called $\Gkdpp$ where $\phi: [k^2] \to [k]\times[k]$ is a bijection used to map any edge $\left(s_{i-1}, s_i\right)$ on the given $s$-$t$ path to an edge on the bipartite graph. Refer to \Cref{fig:Gkdpp} for an illustration. We present the construction of $\Gkdpp$ with some default edge orientation. 
The graph $\Gkdpp$ is constructed via the following steps:
\begin{enumerate}
    \item Construct $\GGdp$ with $\Gamma = 2k$. We relabel the vertices of the last $k$ paths by $w^\ell_i$, where $1\leq\ell\leq k$ and $0\leq i\leq d^p-1$. For $1\leq\ell\leq k$, we use the shorthand $v^\ell$ and $w^\ell$ to denote the \emph{last} vertex of path $\Pc^\ell$ and $\Pc^{\ell+k}$, respectively. 
    \item For $1\leq i,j \leq k$, add edges $\{v^i, w^j\}$.
    \item Add a directed path $\mathcal{P^*}$ with $k^2$ edges from $s$ to $t$.\\ We label the vertices on the path by $s=s_0, s_1, \dots, s_{k^2}=t$. 
    \item Add $k$ paths $\Qc^1, \dots, \Qc^k$ each of length $2k^2$, where $\Qc^\ell$ consists of vertices $q^\ell_0, \dots, q^\ell_{2k^2}$. \\For each $1\leq\ell\leq k$, add edge $\left(q^{\ell}_{2k^2}, v^{\ell}_0\right)$.
    \item Add $k$ paths $\Rc^1, \dots, \Rc^k$ each of length $2k^2$, where $\Rc^\ell$ consists of vertices $r^\ell_0, \dots, r^\ell_{2k^2}$. \\For each $1\leq\ell\leq k$, add edge $\left(w^{\ell}_0, r^{\ell}_{0}\right)$.
    \item For $i\in [k^2]$, suppose $\phi(i)=\left(\phi_1(i),\phi_2(i)\right)$, add $\left(s_{i-1}, q^{\phi_1(i)}_{2(i-1)}\right)$ and $\left(r^{\phi_2(i)}_{2i}, s_i\right)$.
    \item Add edges from $\alpha$ to all vertices in $\Pstar, \Qc^i, \Rc^i$, for each $1\leq i \leq k$.
\end{enumerate}

\paragraph{Direction and removal of edges.}
To facilitate our reduction from set disjointness to \secsisp{}, we assigned default orientations for all edges in $\Gkdpp$ except for those in the right bipartite graph and we remove some of the edges of the form $\left(s_{i-1}, q^{\phi_1(i)}_{2(i-1)}\right)$. The direction of the edges on the bipartite graph is determined by a matrix $M\in \bin^{k\times k}$ and the removal of the edges $\left(s_{i-1}, q^{\phi_1(i)}_{2(i-1)}\right)$ are determined by a vector $x\in \bin^{k^2}$. The directed version of the graph $\GkdppMx$ is constructed as follows.

\begin{enumerate}
    \item Orient the edges in paths $\Pstar, \Pc^\ell, \Qc^\ell, \Rc^\ell$ pointing to larger index for $1\leq \ell \leq k$ and orient the edges in paths $\Pc^\ell$ pointing to smaller index for $k+1\leq \ell \leq 2k$.
    \item Orient the edges in the tree $\Tc$ from parents to children. Orient the edges between the leaves and the paths pointing away from the leaves.  
    \item For $1\leq i,j \leq k$, add edge $(v^i, w^j)$ if $M_{ij}=1$. Otherwise, add edge $(w^j, v^i)$.
    \item For $i\in [k^2]$, keep the edge $\left(s_{i-1}, q^{\phi_1(i)}_{2(i-1)}\right)$ if $x_i = 1$. Otherwise, remove it.
\end{enumerate}

When $x_i = 1$, alternative paths avoiding $\left(s_{i-1}, s_i\right)$ can use the edge $\left(s_{i-1}, q^{\phi_1(i)}_{2(i-1)}\right)$ to make a detour. Conversely, when $x_i = 0$, any alternative paths avoiding $(s_{i-1}, s_i)$ must leave $\Pstar$ before $s_{i-1}$. Since the paths $\Qc^\ell$ have double the length, alternative paths exiting earlier will have longer lengths. If the corresponding edge in the bipartite graph is oriented as $\left(v^{\phi_1(i)}, w^{\phi_2(i)}\right)$ (i.e. when $M_{\phi_1(i), \phi_2(i)}=1$), the returning detour using the edge $\left(r^{\phi_2(i)}_{2i}, s_i\right)$ can be taken. For the same reason, detours returning earlier will have shorter lengths. Hence, the shortest detour is possible if and only if both $x_i = 1$ and $M_{\phi_1(i), \phi_2(i)}=1$, allowing us to reduce set disjointness to \secsisp{}.


\begin{figure}[ht!]
    \centering
    \begin{tikzpicture}[scale=1.1]
\scriptsize
    \tikzset{
        vertex/.style={circle, draw, minimum size=4pt, inner sep=1pt, fill=white},
        filled/.style={circle, draw, fill=black, minimum size=4pt, inner sep=1pt},
        symbol/.style={inner sep=0pt, font=\small, text height=12pt},
        small/.style={font=\small},
        mapping/.style={teal!60!blue, thick},
        dir/.style={arrows = {-Stealth[inset=0.7pt, length=5pt, angle'=25]}},
        widedir/.style={arrows = {-Stealth[inset=0.7pt, length=5pt, angle'=35]}},
        highlight/.style={teal!70!green, thick},
        hlar/.style={highlight, dir},
        line/.style={gray, very thin, dir},
    }
    
    \node[] at (-0.4,0) {$\mathcal{P}^1$};
    \node[vertex] (a) at (0,0) [label=below:$v^1_0$] {};
    \node[vertex] (b) at (1,0) [label=below:$v^1_1$] {};
    \node[vertex] (c) at (2,0) [label=below:$v^1_2$] {};
    \node[] (ca) at (3.5,0) {};
    \node[small] at (4,0) {$\cdots$};
    \node[] (cb) at (4.5,0) {};
    \node[vertex] (d) at (7.5,0) [label=right:{$v^1=v^1_{d^p-1}$}] {};

    \node[symbol] (dda) at (0,-0.7) {$\vdots$};
    \node[symbol] (ddb) at (1,-0.7) {$\vdots$};
    \node[symbol] (ddc) at (2,-0.7) {$\vdots$};
    \node[symbol] (ddd) at (6.6,-0.6) {$\iddots$};
    
    \node[symbol] (ddda) at (0,-3.8) {$\vdots$};
    \node[symbol] (dddb) at (1,-3.8) {$\vdots$};
    \node[symbol] (dddc) at (2,-3.8) {$\vdots$};
    \node[symbol] (dddd) at (6.6,-3.4) {$\ddots$};
    
    \node[] at (-0.4,-1.2) {$\mathcal{P}^k$};
    \node[vertex] (e) at (0,-1.2) [label=below:$v^k_0$] {};
    \node[vertex] (f) at (1,-1.2) [label=below:$v^k_1$] {};
    \node[vertex] (g) at (2,-1.2) [label=below:$v^k_2$] {};
    \node[] (ga) at (3.5,-1.2) {};
    \node[small] at (4,-1.2) {$\cdots$};
    \node[] (gb) at (4.5,-1.2) {};
    \node[vertex] (h) at (6,-1.2) [label=below left:{$v^k$}] {};

    \node[] at (-0.5,-2.8) {$\mathcal{P}^{k+1}$};
    \node[vertex] (e2) at (0,-3) [label=below:$w^1_0$] {};
    \node[vertex] (f2) at (1,-3) [label=below:$w^1_1$] {};
    \node[vertex] (g2) at (2,-3) [label=below:$w^1_2$] {};
    \node[] (ga2) at (3.5,-3) {};
    \node[small] at (4,-3) {$\cdots$};
    \node[] (gb2) at (4.5,-3) {};
    \node[vertex] (h2) at (6,-3) [label=below:{$w^1$}] {};

    \node[] at (-0.5,-4) {$\mathcal{P}^{2k}$};
    \node[vertex] (i) at (0,-4.2) [label=below:$w^k_0$] {};
    \node[vertex] (j) at (1,-4.2) [label=below:$w^k_1$] {};
    \node[vertex] (k) at (2,-4.2) [label=below:$w^k_2$] {};
    \node[] (ka) at (3.5,-4.2) {};
    \node[small] at (4, -4.2) {$\cdots$};
    \node[] (kb) at (4.5,-4.2) {};
    \node[vertex] (l) at (7.5,-4.2) [label=below:{$w^k=w^k_{d^p-1}$}] {};
    
    \begin{scope}[xshift=1.5mm, yshift=0mm]
        \node at (5,3) {$\mathcal{T}$};
        \node[vertex] (o) at (3.3,3.6) [label=right:$u_0^0$] {};
        
        \node (left) at (2.7,3.2) {};
        \node (right) at (4.1,3.0) {};
        
        \node[symbol] (eleft) at (2.4,3.1) {$\iddots$};
        \node[symbol] (eright) at (4.5,2.7) {$\ddots$};
        \node[small] (mid) at (4,1) {$\cdots$};
        \node[vertex] (pq) at (2,2.8) {};
        
        \node[vertex] (p) at (1,2) [label=right:$u^{p-1}_0$] {};
        \node[vertex] (q) at (2.9,2) [label=right:$u^{p-1}_1$] {};
        \node[vertex] (qq) at (5.1,2) {};
        
        \node[filled] (s) at (0.5,) [label=above left:{$\alpha=u^p_0$}] {};
        \node[vertex] (m) at (1.5,1) [label=left:$u^p_1$] {};
        \node[vertex] (n) at (2.5,1) [label=left:$u^p_2$] {};
        \node[vertex] (nn) at (3.3,1) {};
        \node[vertex] (rl) at (4.7,1) {};
        \node[filled] (r) at (5.5,1) [label=above right:{$\beta=u^p_{d^p-1}$}] {};
    \end{scope}
    \draw[dir, highlight] (a) -- (b) -- (c) -- (ca);
    \draw[dir] (e) -- (f) -- (g) -- (ga);
    \draw[dir] (ga2) -- (g2) -- (f2) -- (e2);
    \draw[dir, highlight] (ka) -- (k) -- (j) -- (i);
    \draw[dir, highlight] (cb) -- (d);
    \draw[dir] (gb) -- (h);
    \draw[dir] (h2) -- (gb2);
    \draw[dir, highlight] (l) -- (kb);
    
    \draw[line] (o) -- (left);
    \draw[line] (o) -- (right);
    \draw[line] (pq) -- (p);
    \draw[line] (pq) -- (q);
    \draw[line] (p) -- (s);
    \draw[line] (p) -- (m);
    \draw[line] (q) -- (n);
    \draw[line] (q) -- (nn);
    \draw[line] (qq) -- (r);
    \draw[line] (qq) -- (rl);
    
    \draw[line] (s) to [out=269,in=35] (a);
    \draw[line] (s) to [out=269,in=50] (e);
    \draw[line] (s) to [out=272,in=60] (e2);
    \draw[line] (s) to [out=275,in=60] (i);
    
    \draw[line] (m) to [out=269,in=35] (b);
    \draw[line] (m) to [out=269,in=50] (f);
    \draw[line] (m) to [out=272,in=60] (f2);
    \draw[line] (m) to [out=275,in=60] (j);
    
    \draw[line] (n) to [out=268,in=35] (c);
    \draw[line] (n) to [out=269,in=50] (g);
    \draw[line] (n) to [out=272,in=60] (g2);
    \draw[line] (n) to [out=275,in=60] (k);
    
    \draw[line] (r) to [out=355,in=115] (d);
    \draw[line] (r) to [out=300,in=85] (h);
    \draw[line] (r) to [out=260,in=120, looseness=1] (h2);
    \draw[line] (r) to [out=258,in=175, looseness=1.2] (l);

    \node (vl) at (6.95,-0.25) {};
    \node (wl) at (6.95,-3.85) {};
    \draw[dir, mapping] (d) -- (h2);
    \draw[dir, mapping] (d) -- (wl);
    \draw[dir, mapping] (vl) -- (l);
    \draw[dir, mapping] (wl) -- (vl);
    \draw[dir, mapping] (vl) -- (h2);
    \draw[dir, mapping] (h) -- (l);
    \draw[dir, mapping] (h2) -- (h);
    \draw[dir, mapping] (wl) -- (h);
    \draw[widedir, highlight, very thick] (d) -- (l);

    \begin{scope}[xshift=-35mm, yshift=-1mm]
        \node[vertex] (s0) at (0,-4) [label=left:{$s=s_0$}] {};
        \node[vertex] (s1) at (0,-3) [label=left:$s_1$] {};
        \node[symbol] (s2) at (0,-2) {$\vdots$};
        \node[vertex] (sj) at (0,-1) [label=below left:$s_{i-1}$] {};
        \node[vertex] (si) at (0,0) [label=below left:$s_i$] {};
        \node[symbol] (s5) at (0,1) {$\vdots$};
        \node[vertex] (st) at (0,2) [label=below left:{$t=s_{k^2}$}] {};
        \node[symbol] (Pstar) at (-0.1,2.52)  {$\mathcal{P}^*$};

        \draw[dir, highlight] (s0) -- (s1);
        \draw[dir, highlight] (s1) -- (s2);
        \draw[dir, highlight] (s2) -- (sj);
        \draw[dir, red, thick] (sj) -- (si);
        \draw[dir, highlight] (si) -- (s5);
        \draw[dir, highlight] (s5) -- (st);

        \node[symbol] (Q)  at (1.2,2.5)    {$\mathcal{Q}^1$};
        \node[vertex] (q0)  at (1.3,-4)     [label=right:{$q^1_0$}] {};
        \node[vertex] (q05) at (1.3,-3.5) [label=right:{$q^1_1$}] {};
        \node[vertex] (q1)  at (1.3,-3)    [label=right:{$q^1_2$}] {};
        \node[symbol] (q2)  at (1.3,-2) {$\vdots$};
        \node[vertex] (qj)  at (1.3,-1)    [label=below right:$q^1_{2i-2}$]{};
        \node[vertex] (qji) at (1.3,-0.5) [label=right:$q^1_{2i-1}$]{};
        \node[vertex] (qi)  at (1.3,0)     [label=right:$q^1_{2i}$] {};
        \node[symbol] (q5)  at (1.3,1) {$\vdots$};
        \node[vertex] (qt)  at (1.3,2)     [label=below right:{$q^1_{2k^2}$}] {};

        \draw[dir] (q0)  -- (q05);
        \draw[dir] (q05) -- (q1);
        \draw[dir] (q1)  -- (q2);
        \draw[dir] (q2)  -- (qj);
        \draw[dir, highlight] (qj)  -- (qji);
        \draw[dir, highlight] (qji) -- (qi);
        \draw[dir, highlight] (qi)  -- (q5);
        \draw[dir, highlight] (q5)  -- (qt);

        \node[symbol] (r)  at (-1.5,2.55)    {$\mathcal{R}^k$};
        \node[vertex] (r0)  at (-1.5,-4)    [label=left:{$r^k_0$}] {};
        \node[vertex] (r05) at (-1.5,-3.5)  [label=left:{$r^k_1$}] {};
        \node[vertex] (r1)  at (-1.5,-3)    [label=left:{$r^k_2$}] {};
        \node[symbol] (r2)  at (-1.5,-2) {$\vdots$};
        \node[vertex] (rj)  at (-1.5,-1)    [label=left:$r^k_{2i-2}$]{};
        \node[vertex] (rji) at (-1.5,-0.5)  [label=left:$r^k_{2i-1}$]{};
        \node[vertex] (ri)  at (-1.5,0)     [label=above left:$r^k_{2i}$] {};
        \node[symbol] (r5)  at (-1.5,1) {$\vdots$};
        \node[vertex] (rt)  at (-1.5,2)     [label=below left:{$r^k_{2k^2}$}] {};

        \draw[dir, highlight] (r0)  -- (r05);
        \draw[dir, highlight] (r05) -- (r1);
        \draw[dir, highlight] (r1)  -- (r2);
        \draw[dir, highlight] (r2)  -- (rj);
        \draw[dir, highlight] (rj)  -- (rji);
        \draw[dir, highlight] (rji) -- (ri);
        \draw[dir] (ri)  -- (r5);
        \draw[dir] (r5)  -- (rt);
        
        \draw[hlar] (ri)  to [out=20,in=165] (si);
        \draw[hlar] (sj)  to [out=345,in=200] (qj);
    \end{scope}

    \draw[dir, highlight] (qt)  to [out=30,in=120] (a);
    \draw[dir, highlight] (i)  to [out=200,in=320] (r0);
    \node[symbol] (PPP)  at (-1.5,-5) {{\color{teal!70!green}$P$}};
\end{tikzpicture}
    \caption{Illustration of $\Gkdpp$, assuming $\phi(i)=(1,k)$ and $M_{1,k} = 1$. Paths $\Qc^2,\dots,\Qc^k$ and $\Rc^1,\dots,\Rc^{k-1}$ are omitted. Other edges from $\Pstar$ to $\Qc^1$, and from $\Rc^k$ to $\Pstar$ are omitted. Edges from $\alpha$ to $\Pstar, \Qc^i$ and $\Rc^i$ are omitted. When $(v^1, w^k)$ is oriented from top to bottom and when $\left(s_{i-1}, q^{\phi_1(i)}_{2(i-1)}\right)$ is present, the green highlighted path $P$ is a replacement path from $s$ to $t$ for the deleted edge $(s_{i-1}, s_i)$.}
    \label{fig:Gkdpp}
\end{figure}

\begin{observation}[Basic properties of $\Gkdpp$]\label{obs:Gkdpp}
    There are $\Theta(k^3 + k d^p)$ vertices in $\Gkdpp$, and the diameter of $\Gkdpp$ is at most $2p+2$.
\end{observation}

\begin{proof}
    The number of vertices for the respective parts are: $2k d^p$ vertices for $\Pc^\ell$, $2k \times (2k^2+1)$ for $\Rc^\ell$ and $\Qc^\ell$, $k^2+1$ vertices for $\Pstar$ and $\frac{d^{p+1}-1}{d-1}$ vertices for $\Tc$. Hence the total number of vertices is $2k d^p +4k^3 + 2k + k^2 + 1 + \frac{d^{p+1}-1}{d-1} = \Theta(k^3+kd^p)$. Every vertex not on the tree $\Tc$ is connected to some leaf in $\Tc$. Since $\Tc$ has depth $p$, the diameter of $\Gkdpp$ is at most $2p+2$.
\end{proof}
    
Now we show a lemma for a lower bound of distributed algorithms for set disjointness on $\Gkdpp$. This is analogous to \Cref{lem:disj}. We achieve this by applying \Cref{thm:simulation} and \Cref{lem:disj} in a black-box manner, and showing how distributed algorithms for set disjointness on $\Gkdpp$ can be simulated on $\GGdp$ with $O(1)$-factor overhead. The intuition for the simulation is that $\alpha$ will simulate the vertices in $\Pstar, \Qc^\ell$ and $\Rc^\ell$ and $\beta$ will simulate the vertices in the bipartite graph.


    

\begin{lemma}[Set disjointness lower bound for $\Gkdpp$]
\label{lem:disj G(kdpp)}
    For any $k, d, p$ and bijection $\phi: [k^2]\to [k]\times [k]$ there exists a constant $\epsilon>0$ such that
    $$R_\epsilon^{\Gkdpp}\left(\disj_b;B\right)=\Omega\left(\min\left(d^p, \frac{b}{dpB}\right)\right).$$
\end{lemma}

\begin{proof}
    First, we show that if there is a $T$-round distributed algorithm for set disjointness $\Ac$ on $\Gkdpp$, then there is a $3T$-round distributed algorithm for set disjointness on $\GGdp$. We achieve this via simulating $\Ac$ on $\GGdp$.

    Suppose we have an algorithm $\Ac$ on graph $\GkdppMx$, we will denote a message sent in the $i$th round of $\Ac$ from vertex $u$ to vertex $v$ as $\Mc_i(u\to v)$. 
    Now $\alpha$ in $\GGdp$ can simulate all the vertices in $\Rc^\ell, \Qc^\ell$ and $\Pstar$ in the following way. 
    vertex $\alpha$ will run $2k+2$ processes and run $\Ac$ for $\alpha, \Rc^\ell, \Qc^\ell$ and $\Pstar$ for $1\leq \ell \leq k$ on these processes. 
    Any messages among the vertices in $\alpha, \Rc^\ell, \Qc^\ell$ and $\Pstar$ can be handled locally between the processes.
    The only messages left are between $v^\ell_0$ and $q^\ell_{2k^2}$ and between $w^\ell_0$ and $r^\ell_0$, where $1\leq \ell\leq k$. For each $1\leq \ell\leq k$, $\Mc_i(q^\ell_{2k^2}\to v^\ell_0)$ and $\Mc_i(v^\ell_0 \to q^\ell_{2k^2})$ are delivered along the edge $\{v^\ell_0, \alpha\}$. Similarly, $\Mc_i(r^\ell_{0}\to w^\ell_0)$ and $\Mc_i(w^\ell_0 \to r^\ell_{0})$ are delivered along the edge $\{w^\ell_0, \alpha\}$.  
    Since there are at most twice the number of messages along the edges of $\alpha$, we can simulate $\Ac$ on $\GGdp$ in $3T$ rounds.

    Now we describe how to handle the bipartite graph in $\GkdppMx$ in vertex $\beta$ on $\GGdp$. vertex $\beta$ will create $2k+1$ processes that run $\Ac$ acting as $\beta, v^\ell$ and $w^\ell$ locally, where $1\leq \ell\leq k$. Any messages among these vertices are handled locally on $\beta$. The only messages left are between $v^\ell_{d^p-2}$ and $v^\ell$ and between $w^\ell_{d^p-2}$ and $w^\ell$. We will focus on $v^\ell$ now since the behaviors on $w^\ell$ are similar. In the $(3i-2)$th round, vertex $v^\ell_{d^p-2}$ and $v^\ell$ exchange $\Mc_i(v^\ell_{d^p-2}\to v^\ell)$ and $\Mc_i(v^\ell\to v^\ell_{d^p-2})$. In the $(3i-1)$th round, $v^\ell$ relay $\Mc_i(v^\ell_{d^p-2}\to v^\ell)$ to $\beta$. $\beta$ now has all the information needed to compute $\Mc_{i+1}(v^\ell \to v^\ell_{d^p-2})$ locally. $\beta$ computes $\Mc_{i+1}(v^\ell \to v^\ell_{d^p-2})$ and send it to $v^\ell$ in the $3i$th round. Now $v^\ell$ has $\Mc_{i+1}(v^\ell \to v^\ell_{d^p-2})$ to send to $v^\ell_{d^p-2}$ in the $(3(i+1)-2)$th round and the process repeats.

    All other vertices in $\GGdp$ will execute their part of $\Ac$ with $2$ empty rounds between any two rounds in $\Ac$ to align their pace with that of vertices $\alpha$ and $\beta$.
    In this way, we can simulate any $T$-round algorithm on $\GkdppMx$ with $3T$ rounds on $\GGdp$.
    
    After applying \Cref{lem:disj}, we have the desired lower bound.
\end{proof}

\subsection{Lower Bound for the \secsisp{} Problem}
\label{subsec:lower bound reduction}
Now we are ready to show our main results: lower bounds for the \secsisp{} problem and the \rpath{} problem. We will show \Cref{thm:2sisp lower} by a reduction from the disjointness problem on $\Gkdpp$.
Before we show a reduction from disjointness to \secsisp{}, we show a correspondence between the replacement path lengths and the edge orientations in the bipartite graph.

\begin{lemma}[Replacement path lengths]
\label{lem:GkdppMx Rpath-dir correspondence}
    Consider the graph $\GkdppMx$. For any edge $(s_{i-1},s_i)$ in the path $\Pstar$, if $M_{\phi_1(i), \phi_2(i)}=1$ and $x_i=1$, the length of the replacement path is $3k^2 + 2 d^p + 6$. Otherwise, the length is strictly greater.
\end{lemma}

\begin{proof}
    Suppose $M_{\phi_1(i), \phi_2(i)}=1$ and $x_i = 1$, then the highlighted path $P$ as shown in \Cref{fig:Gkdpp} is an alternative path. It can be easily checked that it has length $3k^2 + 2 d^p + 6$. Now we show that it is the shortest among all alternative $s$-$t$ path. Observe that all alternative $s$-$t$ path must be of the form 
    $$s, \dots, s_{j-1}, q^{\phi_1(j)}_{2j-2}, \dots, q^{\phi_1(j)}_{2k^2}, \;\Pc^{\phi_1(j)},\; \Pc^{\phi_2(l)+k}, \;r^{\phi_2(l)}_0, \dots, r^{\phi_2(l)}_{2l}, s_l,\dots, t$$ where $j\leq i$ and $l\geq i$. 
    Notice that the length of the above path is $3k^2 + 2d^p + 4 + 2(l-j+1)$. This length is minimized when $l=j=i$ and we have our highlight path $P$ with length $3k^2 + 2 d^p + 6$.

    Conversely, suppose $M_{\phi_1(i), \phi_2(i)}=0$ or $x_i=0$. If $M_{\phi_1(i), \phi_2(i)}=0$, the shortest alternative path is has length greater than $3k^2 + 2 d^p + 6$. This is because, the choice of $j=l=i$ does not constitute a directed path anymore as the bipartite edge is oriented from $w^{\phi_2(i)}$ to $v^{\phi_1(i)}$. If $x_i=0$, then the alternate $s$-$t$ path must exit $\Pstar$ before $s_{i-1}$ and therefore has length greater than $3k^2 + 2 d^p + 6$.
\end{proof}

\begin{lemma}[Reducing set disjointness to \secsisp{}]\label{lem:Disj to 2sisp}
    For any $k$, $d\geq 2$, $p$, bijection $\phi: [k^2]\to [k]\times [k]$ and $\epsilon\geq 0$, if there exists an $\epsilon$-error randomized distributed algorithm for the \secsisp{} problem on graph $\GkdppMx$ for any $M\in \bin^{k\times k}$ and any $x\in \bin^{k^2}$ then there exists an $\epsilon$-error randomized algorithm for computing $\disj_{k^2}(x,y)$ (on $k^2$-bit inputs) on $\Gkdpp$ that uses the same round complexity with additional $O(\frac{k}{B})$ rounds.
\end{lemma}

\begin{proof}
    Suppose $\Ac$ is an $\epsilon$-error randomized distributed algorithm for the \secsisp{} problem and suppose we are given a set disjointness instance with $k^2$-bit strings on $\Gkdpp$. We aim to use $\Ac$ to solve the set disjointness problem $\disj_{k^2}(x,y)$.

    Suppose Alice received $x\in \bin^{k^2}$ on vertex $\alpha$ and Bob received $y\in \bin^{k^2}$ on vertex $\beta$. By viewing $y$ as a matrix $M$ (using the lexicographical map) and viewing input $x$ as the argument $x$ in $\GkdppMx$, Alice may remove some of the edges of the form $\left(s_{i-1}, q^{\phi_1(i)}_{2(i-1)}\right)$, and Bob will orient the edges in the bipartite graph to make $\GkdppMx$. Alice only needs to send one bit to each vertex in $\Pstar$. Bob will need to send $k$ bits of information to each of the $2k$ vertices, $v^1,\dots, v^k, w^1,\dots, w^k$. This will take $O(k/B)$ rounds. 
    
   Alice and Bob, together with other vertices, run $\Ac$ on $\GkdppMx$. By \Cref{lem:GkdppMx Rpath-dir correspondence}, the length of the replacement path for each edge $(s_{i-1},s_i)$ is $3k^2 + 2 d^p + 6$ if $M_{\phi_1(i), \phi_2(i)}=1$ and $x_i=1$ and longer otherwise. Since the length of the second simple shortest path is the minimum replacement path across all edges on the shortest path, the length of the second simple shortest path is $3k^2 + 2 d^p + 6$ if and only if there is an index $i\in [k^2]$ such that $M_{\phi_1(i), \phi_2(i)}=1$ and $x_i=1$. Hence, Alice and Bob will output $0$ for $\disj_{k^2}(x,y)$ if and only if the length of the second simple shortest path is $3k^2 + 2 d^p + 6$. The failure probability follows directly.
\end{proof}

\rpathlower*

\begin{proof}[Proof of \Cref{thm:2sisp lower}]
    \Cref{lem:disj G(kdpp)} provides a lower bound for distributed algorithms for computing disjointness on $\Gkdpp$ with input vertices $\alpha$ and $\beta$. Applying the reduction with additive $O(\frac{k}{B})$ overhead from disjointness to \secsisp{} in \Cref{lem:Disj to 2sisp}, we know that for any $k,d,p$, there exists a constant $\epsilon>0$ such that any algorithm for the \secsisp{} problem requires $$\Omega\Bigg(\min\left(d^p, \frac{k^2}{dpB}\right) - \frac{k}{B}\Bigg)$$ rounds in the $\CONGEST(B)$ model on some $\Theta(k^3 + k d^p)$-vertex graph with diameter $2p+2$, by \Cref{obs:Gkdpp}. 


    We set $k^2=d^p$ and rewrite the lower bound in terms of the number of vertices in $\Gkdpp$, $n = \Theta\left(k^3 + k d^p\right) = \Theta\left((d^p)^{3/2}\right)$, as follows:
    \[\Omega\left(\min\left(d^p, \frac{k^2}{dpB}\right)  - \frac{k}{B}\right) = \Omega\left(\frac{k^{2-2/p}}{pB} - \frac{k}{B}\right)= \Omega\left(\frac{(n/2)^{(2/3)(1-1/p)}}{pB}\right)= \Omega\left(\frac{n^{2/3}}{B\log n}\right).\qedhere\]
\end{proof}

\section{Approximation Algorithm for Weighted Directed Graphs}
\label{sec:wapprox}

In this section, we show an $\widetilde{O}(n^{2/3}+D)$-round randomized algorithm for the $(1+\epsilon)$-\apxrpath{} problem in weighted directed graphs. We assume that all the edge weights are positive integers in the range $[W]$, where $W = \poly(n)$ is polynomial in $n$. 


As mentioned earlier, we can handle long-detour replacement paths similarly to the unweighted case. Thus, our focus will primarily be on short-detour replacement paths. We aim to prove the following result, where we still use $\zeta = n^{2/3}$ as the threshold.

\begin{proposition}[Short detours]\label{approx_short}
    For weighted directed graphs, for any constant $\epsilon \in (0,1)$, there exists an $\widetilde{O}(n^{2/3}+D)$-round deterministic algorithm that lets the first endpoint $v_i$ of each edge $e=(v_i, v_{i+1})$ in $P$ compute a number $x$ such that
    \[|st \diamond e| \leq x \leq (1+\epsilon) \cdot \text{the shortest replacement path length for $e$ with a short detour}.\] 
\end{proposition}

Similar to the statement of \Cref{Thm: Long Detour Part Works}, \Cref{approx_short} guarantees an $(1+\epsilon)$ approximation of the value of $|st \diamond e|$ only when some shortest replacement path for $e$ takes a short detour. Otherwise, \Cref{approx_short} provides only an upper bound on $|st \diamond e|$. 

\paragraph{Notations.} We further generalize the definition of $X[\cdot,\cdot]$ in \Cref{sec:short} to capture any two specified sets for possible starting points and ending points. Let $A$ and $B$ be two non-empty sets of integers such that $\max A < \min B$, define $X(A,B)$ as the length of a shortest replacement path with a short detour that starts from a vertex $v_i$ in $P$ with $i \in A$ and ends at a vertex $v_j$ in $P$ with $j \in B$. We define $Y(A,B)$ in the same way but without the requirement that the detour is short. Given a parameter $\epsilon \in (0,1)$, we say that a number $x$ is a \emph{good approximation} of $X(A,B)$ if it satisfies
\[Y(A,B) \leq x \leq (1+\epsilon)X(A,B).\]
For notational simplicity, we may write $\widetilde{X}(A,B)$ to denote a good approximation of $X(A,B)$.
Using the above terminology, the goal of \Cref{approx_short} is to let each vertex $v_i$ compute a good approximation of $X((-\infty,i], [i+1, \infty))$.

\paragraph{Hop-constrained BFS.} Recall that, for the \emph{unweighted} case, by doing a $\zeta$-hop \emph{backward} BFS from each vertex $v_i$ in the path $P$, with some branches trimmed, in $O(\zeta)$ rounds we can let each vertex $v_i$ calculate the \emph{exact} value of $X(\{i\}, [j, \infty))$ for all $j > i$.

A natural approach to extending this method to the \emph{weighted} case is to replace hop-constrained BFS with hop-constrained shortest paths. However, we face a \emph{severe} congestion issue: In the worst case, an edge $e = (u, v)$ may appear in a short detour from $v_i$ to $v_j$ for every pair $(i,j)$ such that $0 \leq i < j \leq h_{st}$. In contrast, when $G$ is \emph{unweighted}, the endpoints of any short detour involving a fixed edge $e$ must lie within an $O(\zeta)$-vertex subpath of $P$. 


To overcome this issue, we apply a \emph{rounding} technique to reduce the weighted case to the unweighted case, so that we can directly apply the same backward BFS algorithm of \Cref{Thm: Backward BFS Works} to let each vertex $v_i$ obtain a \emph{good approximation} of $X(\{i\}, [j, \infty))$ for all $j > i$, and similarly $X((-\infty, j], \{i\})$ for all $j < i$.

\paragraph{Information pipelining.} After the above step,  we want to propagate the information in the path $P$ to enable each vertex $v_i$ to obtain a good approximation of $X((-\infty,i], [i+1, \infty))$. The main issue that we have to deal with here is that short detour paths can now potentially go really far along $P$ in terms of the number of hops. For example, we could have a single edge $(s,t)$ whose weight is one more than the weight of the entire path $P$, and this counts as a short detour path! Unlike the weighted case, here pipelining the information for $\zeta - 1$ hops forwards along $P$ is insufficient to solve the replacement path problem. We fix this by  dividing $\{0,1, \ldots, h_{st}\}$ up into $\ell =O(n^{1/3})$ intervals of $O(n^{2/3})$ indices: \[I_1=[l_1,r_1], I_2=[l_2,r_2], \ldots  I_\ell=[l_\ell,r_\ell],\] 
where $l_1 = 0$ and $r_\ell = h_{st}$. For example, such a partitioning can be obtained by $r_i = \min\{i \cdot \lceil n^{2/3} \rceil, h_{st}\}$ and $l_{i+1} = r_i+1$ for all $i \geq 1$.

Now consider a specific edge $e=(v_i, v_{i+1})$ within an interval $I_j$ (i.e., $\{i, i+1\} \subseteq I_j$). There are two possible cases for a detour path for $e$.
\begin{description}
    \item[Nearby detours:] At least one endpoint $v_k$ of the detour lies in $I_j$ (i.e., $k \in I_j$).
    \item[Distant detours:] Both endpoints $v_k$ and $v_l$ of the detour lie outside $I_j$ (i.e., $k \notin I_j$ and $l \notin I_j$).
\end{description}
The first case can be handled by doing an information pipelining within the interval $I_j$ in $O(n^{2/3})$ rounds. For the second case, we let each interval $I_j$ compute a good approximation of $X( I_j, [l_k, \infty))$ for all $k$ such that $j < k \leq \ell$ and broadcast the result to the entire graph. This provides enough information to not only handle the second case but also cover the short detours for the edges $e$ that crosses two intervals.


\subsection{Approximating Short Detours via Rounding}\label{subsect:rounding}

Given a parameter $\epsilon \in (0,1)$, we say that a collection $C$ of pairs $(j,d)$ is a \emph{short-detour approximator} for a vertex $v_i$ if the following conditions are satisfied.
\begin{description}
    \item[Validity:] Each pair $(j,d) \in C$ satisfies the following requirements:
    \begin{itemize}
        \item $v_j$ is after $v_i$. In other words, $i < j \leq h_{st}$.
        \item $d$ is an upper bound on the shortest replacement path length with a detour starting from $v_i$ and ending at $v_j$. In other words, $d \geq Y(\{i\},\{j\})$. 
        \end{itemize}
    \item[Approximation:] For each $j$ such that $i < j \leq h_{st}$, there exists a  pair $(k,d) \in C$ with $k \geq j$ such that $d \leq (1+\epsilon) \cdot X(\{i\},\{j\})$.
\end{description}

The following lemma explains the purpose of the above definition.

\begin{lemma}[The use of short-detour approximators] A good approximation of $X(\{i\}, [j, \infty))$, for all $j > i$, can be obtained from a short-detour approximator $C$ for vertex $v_i$. \label{lem:obtaining_approx}
\end{lemma}
\begin{proof}
  We select $\widetilde{X}(\{i\}, [j, \infty))$ to be the minimum value of $d$ among all pairs $(k,d) \in C$ with $k \geq j$. By the validity guarantee, we know that 
  \[\widetilde{X}(\{i\}, [j, \infty)) = d \geq Y(\{i\},\{k\}) \geq Y(\{i\}, [j, \infty)).\]
Select $j^\ast \geq j$ to be an index such that there exists a replacement path of length $X(\{i\}, [j, \infty))$ with a short detour starting from $v_i$ and ending at $v_{j^\ast}$. Therefore, $X(\{i\}, [j, \infty)) = X(\{i\},\{j^\ast\})$.
By the approximation guarantee, there exists a  pair $(k^\ast,d^\ast) \in C$ with $k^\ast \geq j^\ast \geq j$ such that $d^\ast \leq (1+\epsilon) \cdot X(\{i\},\{j^\ast\}) = (1+\epsilon) \cdot X(\{i\}, [j, \infty))$. Our choice of $\widetilde{X}(\{i\}, [j, \infty))$ guarantees that 
\[\widetilde{X}(\{i\}, [j, \infty)) \leq d^\ast \leq (1+\epsilon) \cdot X(\{i\}, [j, \infty)).\]
Therefore, $\widetilde{X}(\{i\}, [j, \infty))$ is a good approximation of $X(\{i\}, [j, \infty))$.
\end{proof}

\paragraph{Rounding.} To compute short-detour approximators for all vertices $v_i$ in $P$, we use a rounding technique. For any number $d  > 0$, we define the graph $G_d$ as the result of the following construction.
\begin{enumerate}
    \item Start from the graph $G \setminus P$.
    \item Set $\mu_d = \frac{\epsilon d}{2\zeta}$ to be the unit for rounding.
    \item Replace each edge $e$ in $G \setminus P$ with a path of $\lceil w(e)/\mu_d\rceil$ edges, each of weight $\mu_d$. Here $w(e)$ is the weight of $e$ in $G$.
\end{enumerate}

We summarize the basic properties of $G_d$ as the following observations.

\begin{observation}[Distances do not shrink]\label{obs1}
For any two vertices $u,v \in V$, \[\dist_{G \setminus P}(u,v) \leq \dist_{G_d}(u,v).\]
\end{observation}
\begin{proof}
    This observation follows from the fact that we only round up the edge weights: Each edge $e$ of weight $w$ in  $G \setminus P$ is replaced with a path of length $\mu_d \cdot \lceil w/\mu_d\rceil \geq w$.
\end{proof}

\begin{observation}[Approximation]\label{obs2}
For any two vertices $u,v \in V$, suppose there is a $u$-$v$ path in $G \setminus P$ of length $d' \in [d/2, d]$ and with at most $\zeta$ hops, then there is a $u$-$v$ path in $G_d$ of length at most $d' \cdot (1+\epsilon)$ and with at most $\zeta(1+2/\epsilon)$ hops.
\end{observation}
\begin{proof}
   We simply take the corresponding path in $G_d$. The number of hops of this path is at most \[\frac{d'}{\mu_d} + \zeta \leq \frac{d}{\mu_d} + \zeta = \zeta(1+2/\epsilon),\]
   where the additive term $+\zeta$ is to capture the term $+1$ in the fact that each edge $e$ of weight $w$ in  $G \setminus P$ is replaced with a path of $\lceil w/\mu_d\rceil \leq (w/\mu_d) + 1$ edges in the construction of $G_d$.

   The length of this path is at most 
   \[d' + \zeta\cdot \mu_d  = d' + \frac{\epsilon d}{2} \leq  d'\cdot (1+\epsilon),\]
      where the additive term $+ \zeta\cdot \mu_d $ is to capture the term $+\mu_d$ in the fact that each edge $e$ of weight $w$ in  $G \setminus P$ is replaced with a path of length $\mu_d \cdot \lceil w/\mu_d\rceil \leq w + \mu_d$ in the construction of $G_d$.
\end{proof}

Now we apply the rounding technique to compute the short-detour approximators.

\begin{lemma}[Computing short-detour approximators]
\label{lem:rounding}
There exists an $\widetilde{O}(n^{2/3})$-round deterministic algorithm that lets each vertex $v_i$ in $P$ compute its short-detour approximator.
\end{lemma}

\begin{proof}
For each $d = 2^1, 2^2, 2^3, \ldots, 2^{\lceil \log (mW) \rceil}$, we run the algorithm of \Cref{Thm: Backward BFS Works} with parameter $\zeta^\ast = \zeta(1+2/\epsilon)$ in the graph $G_d$ by treating $G_d$ as an unweighted undirected graph. Here $2^{\lceil \log (mW) \rceil} = n^{O(1)}$ is an upper bound on any path length in $G$. The procedure takes \[O(\zeta^\ast \cdot \log n) = O((\zeta/\epsilon) \log n)= \widetilde{O}(n^{2/3})\] rounds, as $\zeta = n^{2/3}$, $\epsilon \in (0,1)$ is a constant, and there are $O(\log n)$ many choices of $d$.

Now we focus on one vertex $v_i$ in $P$ and discuss how $v_i$ can compute a desired short-detour approximator $C$.
Each execution of the algorithm of \Cref{Thm: Backward BFS Works} lets 
 $v_i$ compute the value of  $f_{v_i}^\ast(h)$ for each $h \in [\zeta^\ast]$. If $f_{v_i}^\ast(h) \neq -\infty$, then we add $(j, d')$ to $C$ with 
\[ j = f_{v_i}^\ast(h) \ \ \ \text{and} \ \ \ d' = \dist_G(s, v_i) + h\cdot \mu_d + \dist_G(v_{j},t).\]
To let $v_i$ learn $\dist_G(v_{j},t)$, we just need to slightly modify the algorithm of \Cref{Thm: Backward BFS Works} to attach this distance information $\dist_G(v_{j},t)$ to the message containing the index $j$. For the rest of the proof, we show that $C$ is a short-detour approximator for $v_i$. 

\paragraph{Validity.} For the validity requirement, observe that whenever we add $(j, d')$ to $C$, there exists a replacement path of length at most $d'$ with a detour starting from $v_i$ and ending at $v_j$. By the definition of $j = f_{v_i}^\ast(h)$, there exists a path in $G_d$ of $h$ hops from $v_i$ to $v_j$. As the length of this path is $h\cdot \mu_d$, by \Cref{obs1}, we know that there exists a detour in $G$ from $v_i$ to $v_j$ of length at most $h\cdot \mu_d$, so indeed there exists a replacement path of length at most $d' = \dist_G(s, v_i) + h\cdot \mu_d + \dist_G(v_{j},t)$ with a detour starting from $v_i$ and ending at $v_j$.

\paragraph{Approximation.} For the approximation requirement, we need to show that, for each $j^\circ$ such that $i < j^\circ \leq h_{st}$, there exists a pair $(k^\circ,d^\circ) \in C$ with $k^\circ \geq j^\circ$ such that $d^\circ \leq (1+\epsilon) \cdot X(\{i\},\{j^\circ\})$. By the definition of $X(\{i\},\{j^\circ\})$, we know that there is a $v_i$-$v_{j^\circ}$ path in $G \setminus P$ with at most $\zeta$ hops and with length 
\[r \leq X(\{i\},\{j^\circ\}) - (\dist_G(s, v_i) + \dist_G(v_{j^\circ}, t)).\]
Given such a path, select the parameter $d$ such that 
\[d/2 \leq  r   \leq d.\] 
Consider the execution of the algorithm of \Cref{Thm: Backward BFS Works} in the graph $G_d$. Since we know that there is a $v_i$-$v_{j^\circ}$ path in $G \setminus P$ of length $r \in [d/2, d]$ and with at most $\zeta$ hops, by \Cref{obs2}, there is a $v_i$-$v_{j^\circ}$ path in $G_d$ of length at most $(1+\epsilon) \cdot r$ and has $h \leq \zeta^\ast = \zeta(1+2/\epsilon)$ hops. Therefore, we must have $f_{v_i}^\ast(h) \geq j^\circ$, and our algorithm adds $(k^\circ,d^\circ)$ to $C$ with $k^\circ = f_{v_i}^\ast(h) \geq j^\circ$ and
\begin{align*}
    d^\circ &= \dist_G(s, v_i) + h \cdot \mu_d + \dist_G(v_{k^\circ},t)\\
    &\leq \dist_G(s, v_i) + h \cdot \mu_d + \dist_G(v_{j^\circ},t)\\
    &= \dist_G(s, v_i) + (1+\epsilon) \cdot r + \dist_G(v_{j^\circ},t)\\
    &\leq (1+\epsilon)\cdot X(\{i\},\{j^\circ\}),
\end{align*}
as required. 
\end{proof}

We summarize the discussion so far as a lemma.

\begin{lemma}[Knowledge of $v_i$ before information pipelining]
\label{lem:rounding_summary}
There exists an $\widetilde{O}(n^{2/3})$-round deterministic algorithm that lets each vertex $v_i$ in $P$ obtain the following information.
\begin{itemize}
    \item A good approximation of $X(\{i\}, [j, \infty))$ for all $j > i$.
    \item A good approximation of $X((-\infty, j], \{i\})$ for all $j < i$.
\end{itemize}
\end{lemma}

\begin{proof}
To let each vertex $v_i$ in $P$ obtain a good approximation of $X(\{i\}, [j, \infty))$ for all $j > i$, we simply run the $\widetilde{O}(n^{2/3})$-round algorithm of \Cref{lem:rounding} and then apply \Cref{lem:obtaining_approx}. By reversing all edges, in $\widetilde{O}(n^{2/3})$ rounds, using the same algorithm, we can also let each vertex $v_i$ in $P$ obtain a good approximation of $X(\{i\}, [j, \infty))$ for all $j > i$.
\end{proof}

\subsection{Information Pipelining for Short Detours}\label{subsect:near}

In the subsequent discussion, we write $\widetilde{X}(A,B)$ to denote any good approximation of $X(A,B)$.
To prove \Cref{approx_short}, our goal is to let each vertex $v_i$ obtain $\widetilde{X}((-\infty,i], [i+1, \infty))$, given the information learned during the algorithm of \Cref{lem:rounding_summary}. Recall that we divide $\{0,1,\ldots, h_{st}\}$ into $\ell =O(n^{1/3})$ intervals of $O(n^{2/3})$ indices: $I_1=[l_1,r_1], I_2=[l_2,r_2], \ldots  I_\ell=[l_\ell,r_\ell]$. In the following lemma, we handle nearby detours.

\begin{lemma}[Nearby detours]
\label{lem:near}
There exists an $\widetilde{O}(n^{2/3})$-round deterministic algorithm that ensures the following: For each $j \in [\ell]$, for each $i \in I_j \setminus \{r_j\}$, $v_i$ obtains the following information.
\begin{itemize}
    \item A good approximation of $X([l_j, i], [i+1, \infty))$.
    \item A good approximation of $X((-\infty, i], [i+1, r_j])$.
\end{itemize}
\end{lemma}
\begin{proof}
In the preprocessing step, we run the $\widetilde{O}(n^{2/3})$-round algorithm of \Cref{lem:rounding_summary}, which allows each vertex $v_i$ in $P$ to calculate $\widetilde{X}(\{i\}, [x, \infty))$, for all $x \in [j+1, \infty)$.

To let  $v_{i}$ calculate \[\widetilde{X}([l_j, i], [i+1, \infty)) = \min_{k \in [l_j, i]} \widetilde{X}(\{k\}, [i+1, \infty)),\] it suffices to do a left-to-right sweep in the path $(v_{l_j}, \ldots, v_{i})$ to calculate the minimum. The procedure costs $i-l_j-1\leq |I_j|-1$ rounds.  We can run the procedure for all $i \in I_j \setminus \{r_j\}$ in a pipelining fashion using \[(|I_j|-1) + |I_j \setminus \{r_j\}|- 1 = O(n^{2/3})\]
rounds. Therefore, $O(n^{2/3})$ rounds suffice to let $v_i$ compute $\widetilde{X}([l_j, i], [i+1, \infty))$, for all $i \in I_j \setminus \{r_j\}$.

The task of computing $X((-\infty, i], [i+1, r_j])$ for $v_i$ can be done similarly. By symmetry, using the same algorithm, we can let $v_i$ compute $X((-\infty, i-1], [i, r_j])$, for each $i \in I_j \setminus \{l_j\}$. Now observe that the good approximation needed by $v_i$ is stored in $v_{i+1}$, so one additional round of communication suffices to let each $v_i$ obtain its needed information. 
\end{proof}


In the following lemma, we let the right-most vertex $v_{r_j}$ of each interval $I_j$ prepare the information to be broadcast regarding distant detours.

\begin{lemma}[Information to be broadcast]
\label{lem:far1}
There exists an $\widetilde{O}(n^{2/3})$-round deterministic algorithm that ensures the following: For each $j \in [\ell]$, $v_{r_j}$ obtains the following information.
\begin{itemize}
    \item A good approximation of $X(I_j, [l_k, \infty))$, for all $k \in [j+1, \ell]$.
\end{itemize}
\end{lemma}
\begin{proof}
In the preprocessing step, we run the $\widetilde{O}(n^{2/3})$-round algorithm of \Cref{lem:rounding_summary}, which allows each vertex $v_i$ in $P$ to calculate $\widetilde{X}(\{i\}, [l_k, \infty))$, for all $k \in [j+1, \ell]$. To let  $v_{r_j}$ calculate \[\widetilde{X}(I_j, [l_k, \infty)) = \min_{i \in I_j} \widetilde{X}(\{i\}, [l_k, \infty)),\] it suffices to do a left-to-right sweep in the path $(v_{l_j}, \ldots, v_{r_j})$ to calculate the minimum. The procedure costs $|I_j|-1$ rounds.  We can run the procedure for all $k \in [j+1, \ell]$ in a pipelining fashion using \[(|I_j|-1) + (\ell-(j+1)+1) - 1 = O(n^{2/3}) \text{ rounds.} \qedhere\]

\end{proof}

In the following lemma, we handle distant detours.

\begin{lemma}[Distant detours]
\label{lem:far2}
There exists an $\widetilde{O}(n^{2/3} + D)$-round deterministic algorithm that lets each vertex $v_i$ in $P$ obtain the following information.
\begin{itemize}
    \item A good approximation of $X((-\infty,r_j], [l_k, \infty))$, for all $j,k \in [\ell]$ such that $j < k$.
\end{itemize}
\end{lemma}
\begin{proof}
Run the $\widetilde{O}(n^{2/3})$-round algorithm of \Cref{lem:far1}, and then for each $j \in [\ell]$, let $v_{r_j}$ broadcast $\widetilde{X}(I_j, [l_k, \infty))$, for all $k \in [j+1, \ell]$, to the entire graph. The total number of messages to be broadcast is $O(\ell^2) = O(n^{2/3})$, so the broadcast can be done in $O(n^{2/3} + D)$ rounds by \Cref{LP}. After that, each vertex in the graph can locally calculate
\[\widetilde{X}((-\infty,r_j], [l_k, \infty))=\min_{l\in[j]} \widetilde{X}(I_l, [l_k, \infty)). \qedhere\]
\end{proof}

Combining \Cref{lem:rounding_summary} and \Cref{lem:far2}, we are ready to prove \Cref{approx_short}.

\begin{proof}[Proof of \Cref{approx_short}]
 We run the algorithms of \Cref{lem:near} and \Cref{lem:far2}. Consider the first endpoint $v_i$ of each edge $e=(v_i, v_{i+1})$ in $P$. We just need to show that $v_i$ has enough information to compute $\widetilde{X}((-\infty,i], [i+1, \infty))$. 
 
 We first consider the case where $e$ crosses two intervals, i.e., $i = r_j$ and ${i+1} = l_{j+1}$ for some $j$. In this case, the output of the algorithm of \Cref{lem:far2} already includes \[\widetilde{X}((-\infty,i], [i+1, \infty)) = \widetilde{X}((-\infty,r_j], [l_{j+1}, \infty)).\]

 Next, consider the case where $e$ belongs to one interval, i.e., $\{i, i+1\} \subseteq I_j$ for some $j$. If $I_j$ is the first interval, i.e., $j=1$, then from the output of the algorithm of \Cref{lem:near}, $v_i$ knows \[\widetilde{X}((-\infty, i], [i+1, \infty)) = \widetilde{X}([l_1, i], [i+1, \infty)).\]
If $I_j$ is the last interval, i.e., $j=\ell$, then from the output of the algorithm of \Cref{lem:near}, $v_i$ knows \[\widetilde{X}((-\infty, i], [i+1, \infty)) = \widetilde{X}((-\infty, i], [i+1, r_\ell]).\] 
If $1 < j < \ell$, then $v_i$ can combine the outputs from  \Cref{lem:near} and \Cref{lem:far2} to obtain
\begin{align*}
&\widetilde{X}((-\infty, i], [i+1, \infty))\\ &= \min\left\{\widetilde{X}([l_j, i], [i+1, \infty)), \widetilde{X}((-\infty, i], [i+1, r_j]), \widetilde{X}((-\infty,r_{j-1}], [l_{j+1}, \infty))\right\}. \qedhere
\end{align*}
\end{proof}



\subsection{Long Detours}

For replacement paths with a long detour, they can be handled  by a small modification to the proof of \Cref{Thm: Long Detour Part Works}.
In unweighted graphs, we can compute the $k$-source $h$-hop shortest paths \emph{exactly} by growing the BFS tree from $k$ sources of $h$ hops in $O(k+h)$ rounds using \Cref{kbfs}. For weighted graphs, we use the following result by Nanongkai~\cite[Theorem 3.6]{nanongkai2014distributed}. 

\begin{lemma}[$h$-hop $k$-source shortest paths~\cite{nanongkai2014distributed}]\label{nanongkai2014distributed}
    For any constant $\epsilon \in (0,1)$, there is an $\widetilde{O}(k+h+D)$-round algorithm that approximates the $h$-hop $k$-source shortest paths in a weighted directed graph within a multiplicative factor $(1+\epsilon)$ with high probability.
\end{lemma}

By replacing all BFS computation in our algorithm for long detours in unweighted graphs with the algorithm of \Cref{nanongkai2014distributed}, we obtain the following result.

\begin{proposition}[Long detours]\label{approxsssp}
    For weighted directed graphs, for any constant $\epsilon \in (0,1)$, there exists an $\widetilde{O}(n^{2/3}+D)$-round randomized algorithm that lets the first endpoint $v_i$ of each edge $e=(v_i, v_{i+1})$ in $P$ compute a number $x$ such that
    \[|st \diamond e| \leq x \leq (1+\epsilon) \cdot \text{the shortest replacement path length for $e$ with a long detour}\] with high probability.   
\end{proposition}

\begin{proof}
The proof is almost identical to the proof of \Cref{Thm: Long Detour Part Works}, so here we only highlight the difference. Analogous to \Cref{Thm: Long Detour Part Works}, we need $(1+\epsilon)$-approximate versions of   \Cref{lem : s to land} and \Cref{lem : land to t}. Recall that the proof of \Cref{lem : land to t} is similar to \Cref{lem : s to land}, so $(1+\epsilon)$-approximate version of \Cref{lem : land to t} can also be proved similar to  $(1+\epsilon)$-approximate version of \Cref{lem : s to land}.

To prove \Cref{lem : s to land}, we need  \Cref{lem : land to land} and \Cref{lemma: vilj path}. In \Cref{lem : land to land}, we compute the distances $|l_j l_k|_{G \setminus P}$ via $n^{2/3}$-hop BFS from each $l_j \in L$ and recall that $|L|= \widetilde{O}(n^{1/3})$. We replace the BFS computation with the $(1+\epsilon)$-approximate $n^{2/3}$-hop $\widetilde{O}(n^{1/3})$-source shortest paths algorithm of \Cref{nanongkai2014distributed}. This allows us to  compute an $(1+\epsilon)$ approximation of those $|l_j l_k|_{G \setminus P}$ distances in $\widetilde{O}(n^{2/3}+D)$ rounds. We do a similar modification to \Cref{lemma: vilj path} to let each $v_i$ in $P$ compute an $(1+\epsilon)$-approximation of $|v_il_j|_{G \setminus P}$ for all $l_j \in L$ in $\widetilde{O}(n^{2/3}+D)$ rounds.



    Combining the $(1+\epsilon)$-approximate versions of  \Cref{lem : s to land} and \Cref{lem : land to t}, the first endpoint $v_i$ of each  edge $e=(v_i, v_{i+1})$ in $P$ can compute a number $x$ such that $|st \diamond e| \leq x \leq (1+\epsilon) \cdot \text{the shortest replacement path length for $e$ with a long detour}.$
\end{proof}

Now, we can prove the main result of the section.

\apxUB*
\begin{proof}
Run the algorithms of \Cref{approx_short} and \Cref{approxsssp} with the threshold $\zeta = n^{2/3}$, by taking the minimum of the two outputs, the first endpoint $v_i$ of each edge $e=(v_i, v_{i+1})$ in $P$ correctly computes an $(1+\epsilon)$ approximation of the shortest replacement path length $|st \diamond e|$.
\end{proof}







\section{Conclusions and Open Problems}
\label{sec:conclusions}
In this work, by showing improved upper and lower bounds, we established that $\widetilde{\Theta}(n^{2/3} + D)$ is the \emph{tight} randomized round complexity bound for the replacement path problem in \emph{unweighted} directed graphs. Moreover, our upper bound extends to $(1+\epsilon)$-approximation for \emph{weighted} directed graphs.  Several intriguing questions remain open.

\begin{description}
\item[Weighted undirected graphs:] \citet{manoharan2024computing}  showed an upper bound of $O(\Tsssp + h_{st})$ for \rpath{} in \emph{weighted undirected} graphs, nearly matching their lower bound of $\widetilde{\Omega}(\sqrt{n}+D)$. Here $\Tsssp$ is the round complexity of the SSSP problem in weighted undirected graphs. It is known that $\Tsssp = \widetilde{\Omega}(\sqrt{n}+D)$~\cite{peleg2000near} and $\Tsssp = \widetilde{O}(\sqrt{n}+D) + n^{(2/5) + o(1)} \cdot O(D^{2/5})$~\cite{cao2023parallel}. For small-diameter graphs, the upper and lower bounds are matched up to the additive term $O(h_{st})$. Is it possible to eliminate this term?

    \item[Further lower bounds:] 
    A promising direction for future research is to develop new lower bounds from the lower bound framework of the work~\cite{das2011distributed} that are not restricted to the form $\widetilde{\Omega}(\sqrt{n})$, similar to our $\widetilde{\Omega}(n^{2/3})$ lower bound for the replacement path problem. Some potential problems to consider are as follows. For depth-first search, the state-of-the-art upper bound is $\widetilde{O}(\sqrt{Dn} + n^{3/4})$~\cite{ghaffari_et_al:LIPIcs.DISC.2017.21}, but no lower bound is known. For girth computation, there is still a gap between the upper bound $O(n)$~\cite{holzer2012optimal} and the lower bound $\widetilde{\Omega}(\sqrt{n})$~\cite{frischknecht2012networks}. More generally, for the problem of computing a minimum-weight cycle, the upper and lower bounds are not matched in many scenarios~\cite{manoharan2024computingMinWeightCycleCONGEST}.
    \item[Deterministic algorithms:] In our $\widetilde{O}(n^{2/3} + D)$-round algorithms, the way we compute long-detour replacement paths is inherently randomized. It remains an open question to determine the optimal deterministic complexity for the replacement path problem. Could the derandomization techniques developed in the deterministic APSP algorithms by Agarwal and Ramachandran~\cite{agarwal2018deterministic,agarwal2020faster} be applied to design efficient deterministic algorithms for the replacement path problem?
    \item[Approximation algorithms:] For $(1+\epsilon)$-approximation in weighted directed graphs, there is still a gap between our new upper bound $\widetilde{O}(n^{2/3} + D)$ and the existing lower bound $\widetilde{\Omega}(\sqrt{n} + D)$~\cite{manoharan2024computing}, as our new lower bound $\widetilde{\Omega}(n^{2/3} + D)$ does not extend to this setting. Is it possible to break the upper bound  $\widetilde{O}(n^{2/3} + D)$ for any constant-factor approximation? Even for the \emph{reachability} version of the replacement path problem, we do not know how to break the upper bound $\widetilde{O}(n^{2/3} + D)$.
    \item[Universal optimality:] A \emph{universally optimal} algorithm is one that, for any given input graph, achieves a round complexity matching that of the best algorithm specifically designed for that graph. Recent studies~\cite{universally_optimal_stoc2021,haeupler2022hop} demonstrated that several problems within the complexity class $\widetilde{\Theta}(\sqrt{n} + D)$---such as minimum spanning tree, $(1+\epsilon)$-approximate single-source shortest paths, and $(1+\epsilon)$-approximate minimum cut---admit approximately universally optimal algorithms in the $\CONGEST$ model. Given the similarities between the replacement path problem and these problems, particularly in the techniques used for both upper and lower bounds, is it possible to develop approximately universally optimal algorithms for the replacement path problem?
    \item[Distance sensitivity oracle:] In the \emph{distance sensitivity oracle} problem, the goal is to process the input graph so that future queries about the shortest-path distance from $s$ to $t$ avoiding $e$ can be answered efficiently, for any $s$, $t$, and $e$.
    Very recently, Manoharan and Ramachandran~\cite{manoharan2024distributed} presented the first distributed algorithms for the distance sensitivity oracle problem in the $\CONGEST$ model. Could the techniques introduced in our work help resolve the remaining open questions about the distance sensitivity oracle problem? Notably, a gap still exists between the upper and lower bounds for weighted directed graphs.

\end{description}

\section*{Acknowledgements}
We would like to thank Vignesh Manoharan and Vijaya Ramachandran for helpful discussions and comments.

\printbibliography
\end{document}